\documentclass[11pt,english]{article}
\usepackage[T1]{fontenc}
\usepackage[latin9]{inputenc}
\usepackage{geometry}
\usepackage{color}
\usepackage{array}
\usepackage{multirow}
\usepackage{amsmath}
\usepackage{amssymb}
\usepackage[authoryear]{natbib}

\makeatletter

\providecommand{\tabularnewline}{\\}

\@ifundefined{date}{}{\date{}}

\usepackage{graphics}

\newtheorem{assumption}{Assumption}
\newtheorem{theorem}{Theorem}

\newtheorem{lemma}{Lemma}
\newtheorem{proof}{Proof}

\newtheorem{remark}{Remark}

\newcommand*{\indep}{%
\bot
}
\newcommand*{\nindep}{%
  \mathbin{
    \mathpalette{\@indep}{\not}
  }%
}
\newcommand*{\@indep}[2]{%
  \sbox0{$#1\perp\m@th$}
  \sbox2{$#1=$}
  \sbox4{$#1\vcenter{}$}
  \rlap{\copy0}
  \dimen@=\dimexpr\ht2-\ht4-.2pt\relax
  \kern\dimen@
  {#2}%
  \kern\dimen@
  \copy0 
} 

\newcommand{\pr}{P} 
\newcommand{\var}{\mathrm{var}}

\newcommand{\R}{\mathbb{R}}
\newcommand{\de}{\mathrm{d}}
\newcommand{\T}{\mathrm{\scriptscriptstyle T}}

\newcommand{\opt}{\mathrm{opt}}

\newcommand{\eff}{\mathrm{eff}}
\newcommand{\dr}{\mathrm{dr}}

\newcommand{\N}{\mathcal{N}}

\newcommand{\naive}{\mathrm{naive}}
\newcommand{\ipcw}{\mathrm{ipcw}}

\newcommand{\msm}{\mathrm{msm}}

\newcommand{\disc}{\mathrm{disc}}

\makeatother

\usepackage{babel}
\begin{document}
\title{\textbf{\huge{}Semiparametric estimation of structural failure time
model in continuous-time processes}}
\author{Shu Yang\thanks{Department of Statistics, North Carolina State University, North Carolina
27695, U.S.A. Email: syang24@ncsu.edu}, Karen Pieper\thanks{Duke Clinical Research Institute, North Carolina, U.S.A. Email: karen.pieper@duke.edu},
and Frank Cools\thanks{AZ Klina, Belgium. Email: frank.cools@klina.be}}
\maketitle
\begin{abstract}
Structural failure time models are causal models for estimating the
effect of time-varying treatments on a survival outcome. G-estimation
and artificial censoring have been proposed to estimate the model
parameters in the presence of time-dependent confounding and administrative
censoring. However, most of existing methods require manually preprocessing
data into regularly spaced data, which may invalidate the subsequent
causal analysis. Moreover, the computation and inference are challenging
due to the non-smoothness of artificial censoring. We propose a class
of continuous-time structural failure time models, which respects
the continuous time nature of the underlying data processes. Under
a martingale condition of no unmeasured confounding, we show that
the model parameters are identifiable from potentially infinite estimating
equations. Using the semiparametric efficiency theory, we derive the
first semiparametric doubly robust estimators, in the sense that the
estimators are consistent if either the treatment process model or
the failure time model is correctly specified, but not necessarily
both. Moreover, we propose using inverse probability of censoring
weighting to deal with dependent censoring. In contrast to artificial
censoring, our weighting strategy does not introduce non-smoothness
in estimation and ensures that the resampling methods can be used
to make inference.

\bigskip{}
\end{abstract}

\section{Introduction}

Confounding by indication is common in observational studies, which
obscures the causal relationship of the treatment and outcome \citep{robins1992g}.
In longitudinal observational studies, this phenomenon becomes more
pronounced due to time-varying confounding when there are time-dependent
covariates that predict the subsequent treatment and outcome and also
are affected by the past treatment history. In this case, standard
regression methods whether or not adjusting for confounders are fallible
\citep{robins2000marginalstructural,daniel2013methods}.

Structural failure time models \citep{robins1991correcting,robins1992estimation}
and marginal structural models \citep{robins2000marginal,hernan2001marginal}
have been used to effectively handle time-varying confounding. Structural
failure time models simulate the potential failure time outcome that
would have been observed in the absence of treatment by removing the
effect of treatment, while marginal structural models specify the
marginal relationship of potential outcomes under different treatments
possibly adjusting for the baseline covariates. Structural failure
time models have certain features that are more desirable than marginal
structural models \textcolor{black}{\citep{robins2000marginal}}:
structural failure time models allow for modeling time-varying treatment
modification effects using the post baseline time-dependent covariates;
they are more flexible to translate biological hypotheses into their
parameters \citep{robins1998structural,lok2008statistical}; and the
g-estimation \citep{robins1998structural} for structural failure
time models does not require the probability of receiving treatment
at each time point to be positive for all subjects. 

Most of structural failure time models specify deterministic relationships
of the observed failure time and the baseline failure time and therefore
are rank preserving; see, e.g., \citet{mark1993method,mark1993estimating,robins1994adjusting,robins2002analytic,hernan2005structural}.
Moreover,  existing g-estimation often uses a discrete-time setup,
which requires all subjects to be followed at the same pre-fixed time
points. However, in practical situations, the variables and processes
are more likely to be measured at irregularly spaced time points,
which may not be the same for all subjects \citep{robins1998correction}.
To apply existing estimators, one needs to discretize the timeline
and re-create the measurements at each time point e.g. by averaging
observations within the given time point or imputation if there are
no observations. Such data preprocessing may distort the relationship
of variables and cast doubt on the sequential randomization assumption,
which however is essential to justify the discrete-time g-estimation
\citep{zhang2011causal}. In the literature, much less work has been
addressing non-rank preserving continuous-time causal models; exceptions
include \citet{robins1998structural,lok2004estimating,lok2008statistical,lok2017mimicking}.
\citet{robins1998structural} conjectured that g-estimation extends
to the settings with continuous-time processes, which however relies
on the rank preserving assumption. Recently, \citet{lok2017mimicking}
presented a formal proof for this conjecture without rank preservation.

Despite these advances, estimation for continuous-time structural
failure time models is largely under-developed. Existing g-estimation
is singly robust, in the sense that it relies on a correct model specification
for the treatment process. In the literature of missing data analysis
and causal inference, many authors have proposed doubly robust estimators
that require either one of the two model components to be correctly
specified \citep{robins1994estimation,scharfstein1999adjusting,van2002locally,lunceford2004stratification,bang2005doubly,cao2009improving,Robins2007,lok2012impact}.
\citet{yang2015gof} constructed a doubly robust test procedure for
structural nested mean models. To our best knowledge, there does not
exist a double robust estimator for structural failure time models.

We develop a general framework for structural failure time models
with continuous-time processes. We relax the local rank preservation
by specifying a distributional instead of deterministic relationship
of the treatment process and the potential baseline failure time.
We impose a martingale condition of no unmeasured confounding, which
serves as the basis for identification and estimation. Under the semiparametric
model characterized by the structural failure time model and the no
unmeasured confounding assumption, we develop a class of regular asymptotically
linear estimators. This class of estimators contains the semiparametric
efficient estimators \citep{bickel1993efficient,tsiatis2007semiparametric}.
To ease computation, we further construct an optimal member among
a wide class of semiparametric estimators that are relatively simple
to compute. Moreover, we show that our estimators are doubly robust,
which achieve the consistency if either the model for the treatment
process is correctly specified or the failure time model is correctly
specified, but not necessarily both. Our framework is readily applicable
to the traditional discrete-time settings.

In the presence of censoring, Robins and coauthors have introduced
the notion of the potential censoring time and proposed an approach
for using this information to estimate the treatment effect. This
approach may artificially terminate follow-up for some subjects before
their observed failure or censoring times, and therefore it is often
called artificial censoring. This approach works only for administrative
censoring when follow-up ends at a pre-specified date. It fails to
provide consistent estimators for dependent censoring \citep{rotnitzky1995semiparametric},
which likely occurs due to subjects drop out. Moreover, the computation
and inference are challenging due to the non-smoothness of artificial
censoring \citep{joffe2001administrative,joffe2012g}. To overcome
these limitations, we propose using inverse probability of censoring
weighting to deal with censoring. In contrast to artificial censoring,
our weighting strategy is smooth and ensures that the resampling methods
can be used for inference, which is straightforward to implement in
practice.

\section{Notation, models, and assumptions\label{sec:Notation}}

\subsection{Notation}

We assume that $n$ subjects constitute a random sample from a larger
population of interest and therefore are independent and identically
distributed. For notational simplicity, we suppress the subscript
$i$ for subject. Let $T$ be the observed failure time. Let $L_{t}$
be a multidimensional covariates process, and let $A_{t}$ be the
binary treatment process; i.e., $A_{t}=1$ if the subject is on treatment
at time $t$, and $A_{t}=0$ if the subject is off treatment at time
$t$. We assume that all subjects received treatment at baseline and
may discontinue treatment during follow up. We also assume that treatment
discontinuation is permanent; i.e., if $A_{t}=0$, then $A_{u}=0$
for all $u\geq t$. Let $V$ be the time to treatment discontinuation
or failure, whichever came first, and let $\Gamma$ be the binary
indicator of treatment discontinuation at time $V$. For the purpose
of regularity, we assume that all continuous-time processes are cadlag
processes; i.e., the processes are continuous from the right and have
limits from the left. Let $H_{t}=(L_{t},A_{t-})$ be the combined
covariates and treatment process, where we write $A_{t-}$ for the
treatment just before time $t$. We also use overline to denote the
history; e.g., $\overline{H}_{t}$$=(H_{u}:0\leq u\leq t)$ is the
history of the covariates and treatment process until time $t$. Following
\citet{cox1984oakes}, we assume there exists a potential baseline
failure time $U$, representing the failure time outcome had the treatment
always been withheld. The full data is $F=(T,\overline{H}_{T})$.
We assume that there is no censoring before $T$ until $\mathsection$
\ref{sec:Censoring}.

\subsection{Structural failure time model}

The structural failure time model specifies the relationship of the
potential baseline failure time $U$ and the actual observed failure
time $T$. We assume that given any $\overline{H}_{t}$,
\begin{equation}
U\sim U(\psi^{*})=\int_{0}^{T}\exp[\{\psi_{1}^{*}+\psi_{2}^{*\T}g(L_{u})\}A_{u}]\de u,\label{eq:trt model}
\end{equation}
where $\sim$ means ``has the same distribution as'', and $\psi^{*\T}=(\psi_{1}^{*},\psi_{2}^{*\T})$
is a $p$-vector of unknown parameters. Model (\ref{eq:trt model})
entails that the treatment effect is to accelerate or decelerate the
failure time compared to the baseline failure time $U$. Intuitively,
$\exp[\{\psi_{1}^{*}+\psi_{2}^{*\T}g(L_{t})\}A_{t}]$ can be interpreted
as the effect rate of the treatment on the outcome possibly modified
by the time-varying covariate $g(L_{t})$. To help understanding the
model, consider a simplified model $U(\psi^{*})=\int_{0}^{T}\exp(\psi_{1}^{*}A_{u})\de u$.
The multiplicative factor $\exp(\psi_{1}^{*})$ describes the relative
increase/decrease in the failure time had the subject continuously
received treatment compared to had the treatment always been withheld. 

\begin{remark}

The rank-preserving structural failure time model specifies a deterministic
relationship instead of a distributional relationship of the failure
times; i.e., it uses ``$=$'' instead of ``$\sim$'' in Model
(\ref{eq:trt model}). Then, for subjects $i$ and $j$ who have the
same observed treatment and covariate history, $T_{i}<T_{j}$ must
imply $U_{i}<U_{j}$. This may be restrictive in practice. In contrast,
we link the distribution of the baseline failure time and the distribution\textit{
}of the actual failure time after removing the treatment effect. Specifically,
we assume that the distributions of $U$ and $U(\psi^{*})$ are the
same, given past treatment and covariates, which do not impose the
rank-preserving restriction.

\end{remark}

\subsection{No unmeasured confounding}

The model parameter $\psi^{*}$ is not identifiable in general, because
$U$ is missing for all subjects. To identify and estimate $\psi^{*}$,
we impose the assumption of no unmeasured confounding \citep{yang2018modeling}.

\begin{assumption}[No Unmeasured Confounding]\label{asump:UNC}The
hazard of treatment discontinuation is 
\begin{eqnarray}
\lambda_{V}(t\mid F,U) & = & \lim_{h\rightarrow0}h^{-1}P(t\leq V<t+h,\Gamma=1\mid F,U,V\geq t)\nonumber \\
 & = & \lim_{h\rightarrow0}h^{-1}P(t\leq V<t+h,\Gamma=1\mid\overline{H}_{t},V\geq t)=\lambda_{V}\left(t\mid\overline{H}_{t}\right).\label{eq:UNC}
\end{eqnarray}

\end{assumption}

Assumption \ref{asump:UNC} implies that $\lambda_{V}(t\mid F,U)$
depends only on the past treatment and covariate history until time
$t$, $\overline{H}_{t}$, but not on the future variables and $U$.
This assumption holds if the set of historical covariates contains
all prognostic factors for the failure time that affect the decision
of discontinuing treatment at $t$.

For an equivalent representation of the treatment process $A_{t}$,
we define the counting process $N_{V}(t)=I(V\leq t,\Gamma=1)$ and
the at-risk process $Y_{V}(t)=I(V\geq t)$ \citep{andersen2012statistical}.
Let $\sigma(H_{t})$ be the $\sigma$-field generated by $H_{t}$,
and let $\sigma(\overline{H}_{t})$ be the $\sigma$-field generated
by $\cup_{u\leq t}\sigma(H_{u})$. We show in the supplementary material
that under Model (\ref{eq:trt model}), (\ref{eq:UNC}) implies that
\begin{equation}
\lambda_{V}\{t\mid\overline{H}_{t},U(\psi^{*})\}=\lambda_{V}(t\mid\overline{H}_{t}).\label{eq:UNC2}
\end{equation}
Thus, under common regularity conditions for the counting process,
$M_{V}(t)=N_{V}(t)-\int_{0}^{t}\lambda_{V}(u\mid\overline{H}_{u})Y_{V}(u)\de u$
is a martingale with respect to $\sigma\{U(\psi^{*}),\overline{H}_{t}\}$,
which renders $\psi^{*}$ identifiable as we show in $\mathsection$
\ref{sec:Identification}. We now focus on semiparametric estimation
in the next section.

\section{Semiparametric estimation\label{sec:Semiparametric-estimation}}

We consider the semiparametric model characterized by Model (\ref{eq:trt model})
and Assumption \ref{asump:UNC}. We derive a regular asymptotically
linear estimator $\widehat{\psi}$ of $\psi^{*}$; i.e.
\begin{equation}
n^{1/2}(\widehat{\psi}-\psi^{*})=P_{n}\Phi(F)+o_{P}(1),\label{eq:RAL}
\end{equation}
where $P_{n}$ is the empirical measure induced by $F_{1},\ldots,F_{n}$,
i.e., $P_{n}\Phi(F)=n^{-1}\sum_{i=1}^{n}\Phi(F_{i})$, and $\Phi(F)$
is the influence function of $\widehat{\psi}$, which has mean zero
and finite and non-singular variance. 

Let $f_{F}(T,\overline{H}_{T};\psi,\theta)$ be the semiparametric
likelihood function based on a single variable $F$, where $\psi$
is the primary parameter of interest, and $\theta$ is the infinite-dimension
nuisance parameter under the semiparametric model. A fundamental result
in \citet{bickel1993efficient} states that the influence functions
for regular asymptotically linear estimators lie in the orthogonal
complement of the nuisance tangent space, denoted by $\Lambda^{\bot}$.
We now characterize $\Lambda^{\bot}$ and defer the proof to the supplementary
material.

\begin{theorem}\label{Thm:Ortho}

Under Model (\ref{eq:trt model}) and Assumption \ref{asump:UNC},
the orthogonal complement of nuisance tangent space for $\psi^{*}$
is 
\begin{equation}
\Lambda^{\bot}=\left\{ \int_{0}^{\infty}\left(h_{u}\{U(\psi^{*}),\overline{H}_{u}\}-E\left[h_{u}\{U(\psi^{*}),\overline{H}_{u}\}\mid\overline{H}_{u},V\geq u\right]\right)\de M_{V}(u)\right\} ,\label{eq:orthogonal nuisance}
\end{equation}
for all $p$-dimensional $\text{\ensuremath{h_{u}}}\{U(\psi^{*}),\overline{H}_{u}\}$.

\end{theorem}

Denote the score function of $\psi^{*}$ as $S_{\psi}(F)=\partial\log f_{F}\left(T,\overline{H}_{T};\psi,\theta\right)/\partial\psi$
evaluated at $(\psi^{*},\theta^{*})$. Following \citet{bickel1993efficient},
the efficient score for $\psi^{*}$ is $S_{\eff}(F)=\prod\left\{ S_{\psi}(F)\mid\Lambda^{\bot}\right\} $,
where $\prod$ is the projection operator in the Hilbert space. The
efficient influence function is $\Phi(F)=E\left\{ S_{\eff}(F)S_{\eff}(F)^{\T}\right\} ^{-1}$$\times S_{\eff}(F)$,
with the variance $\left[E\left\{ S_{\eff}(F)S_{\eff}(F)^{\T}\right\} \right]^{-1}$,
which achieves the semiparametric efficiency bound. However, the analytical
form of $S_{\psi}(F)$ is intractable in general. To facilitate estimation,
we focus on a reduced class of $\Lambda^{\bot}$ with $h_{u}\{U(\psi^{*}),\overline{H}_{u}\}=c(\overline{H}_{u})U(\psi^{*})$
for $c(\overline{H}_{u})\in\mathbb{\R}^{p}$, leading to the estimating
function for $\psi^{*}$: 
\begin{equation}
G(\psi;F)=\int_{0}^{\infty}c(\overline{H}_{u})\left[U(\psi)-E\left\{ U(\psi)\mid\overline{H}_{u},V\geq u\right\} \right]\de M_{V}(u).\label{eq:G}
\end{equation}
Because of the no unmeasured confounding assumption, $U(\psi^{*})\indep M_{V}(u)\mid(\overline{H}_{u},V\geq u)$,
and therefore $E\{G(\psi^{*};F)\}=0$. We obtain the estimator of
$\psi^{*}$ by solving 
\begin{equation}
P_{n}\left\{ G(\psi;F)\right\} =0.\label{eq:ee4}
\end{equation}
Within this class, we show that the optimal choice of $c(\overline{H}_{u})$
is
\begin{equation}
c^{\opt}(\overline{H}_{u})=E\left\{ \partial\dot{U}_{u}(\psi)/\partial\psi\mid\overline{H}_{u},V=u\right\} \left[\mathrm{var}\left\{ U(\psi)\mid\overline{H}_{u},V\geq u\right\} \right]^{-1},\label{eq:c-opt}
\end{equation}
in the sense that with this choice the solution to (\ref{eq:ee4})
gives the most precise estimator of $\psi^{*}$ among all solutions
to (\ref{eq:ee4}). To use $c^{\opt}(\overline{H}_{u})$, we require
positing working models for approximation; see the example in the
simulation study. Compared to naive choices, e.g., $c(\overline{H}_{u})=\{A_{u},A_{u}g(L_{u})^{\T}\}^{\T}$
for Model (\ref{eq:trt model}), our simulation results show that
using the optimal choice gains estimation efficiency.

In (\ref{eq:ee4}), we assume that the model for the treatment process
and $E\left\{ U(\psi)\mid\overline{H}_{u},V\geq u\right\} $ are known.
In practice, they are often unknown and must be modeled and estimated
from the data. We posit a proportional hazards model with time-dependent
covariates; i.e.,
\begin{eqnarray}
\lambda_{V}\left(t\mid\overline{H}_{t};\gamma_{V}\right) & = & \lambda_{V,0}(t)\exp\left\{ \gamma_{V}^{\T}g_{V}(t,\overline{H}_{t})\right\} ,\label{eq:ph-V}
\end{eqnarray}
where $\lambda_{V,0}(t)$ is unknown and non-negative, $g_{V}(t,\overline{H}_{t})$
is a pre-specified function of $t$ and $\overline{H}_{t}$, and $\gamma_{V}$
is a vector of unknown parameters. We also posit a working model $E\left\{ U(\psi)\mid\overline{H}_{u},V\geq u;\xi\right\} $,
indexed by $\xi$. We show that the estimating equation for $\psi^{*}$
achieves the double robustness or double protection \citep{rotnitzky2015double}.

\begin{theorem}[Double robustness]\label{Thm:2-dr}Under Model (\ref{eq:trt model})
and Assumption \ref{asump:UNC}, the estimating equation (\ref{eq:ee4})
for $\psi^{*}$ is unbiased of zero if either the model for the treatment
process is correctly specified, or the failure time model $E\left\{ U(\psi)\mid\overline{H}_{u},V\geq u;\xi\right\} $
is correctly specified, but not necessarily both.

\end{theorem}

\section{Censoring\label{sec:Censoring}}

\subsection{Inverse probability of censoring weighting\label{subsec:algorithm}}

In most studies, the failure time is subject to right censoring. We
now introduce $C$ to be the time to censoring. The observed data
are $O=\{X=\min(T,C),\Delta=1(T\leq C),\overline{H}_{X}\}$. In the
presence of censoring, we may not observe $T$ and calculate $U(\psi)$,
and consequently the estimating equation (\ref{eq:ee4}) is not feasible
to solve. A naive solution is to replace $T$ in $U(\psi)$ by $X$
and use $\widetilde{U}(\psi)=\int_{0}^{X}\exp(\psi A_{s})\de s$;
however, $\widetilde{U}(\psi^{*})$ depends on the whole treatment
process and therefore is not independent of $M_{V}(t)$ given $\overline{H}_{t}$,
which renders the estimating equation (\ref{eq:ee4}) biased \citep{hernan2005structural}.
\citet{robins1998structural} proposed a strategy to deal with administrative
censoring. In this case, $C$ is independent of all other variables.
This strategy replaces $U(\psi)$ by a function of $U(\psi)$ and
$C$ which is always observable. For illustration, we consider $U(\psi)=\int_{0}^{T}\exp\left(\psi A_{u}\right)\de u$
and 
\[
C(\psi)=\min_{a_{s}\in\{0,1\}}\int_{0}^{C}\exp\left(\psi a_{s}\right)\de s=\begin{cases}
C, & \text{if }\psi\geq0,\\
C\exp\left(\psi\right), & \text{if }\psi<0.
\end{cases}
\]
Then, $\widetilde{U}(\psi^{*})=\min\{U(\psi^{*}),C(\psi^{*})\}$ and
$\Delta(\psi^{*})=1\{U(\psi^{*})<C(\psi^{*})\}$ are the two functions
that are independent of $M_{V}(t)$ given $\overline{H}_{t}$ and
are always computable; see the supplementary material. G-estimator
is then constructed based on $\widetilde{U}(\psi)$ and $\Delta(\psi)$.
In this approach, for subjects with $T<C$, it may be possible that
$U(\psi)>C(\psi)$ and $\Delta(\psi)=0$, which considers these subjects
who actually were observed to fail as if they were censored. Therefore,
this approach is often called artificial censoring. Artificial censoring
suffers from many drawbacks. First, the resulting estimating equation
is not smooth in $\psi$, and therefore estimation and inference are
challenging \citep{joffe2012g}. Second, if the censoring mechanism
is dependent, the estimators will be inconsistent \citep{robins1998structural}.
To avoid the drawbacks of artificial censoring and also allow for
more general censoring mechanisms, we consider an alternative approach
using inverse probability of censoring weighting. \citet{robins1998structural}
suggested and \citet{witteman1998g} applied the weighting approach
to deal with censoring by competing risks in the deterministic structural
nested failure time models with discretized data. 

We assume an ignorable censoring mechanism as follows.

\begin{assumption}\label{asp:NUC-1}The hazard of censoring is 
\begin{eqnarray}
\lambda_{C}(t\mid F,T>t) & = & \lim_{h\rightarrow0}h^{-1}P(t\leq C<t+h\mid C\geq t,F,T>t)\nonumber \\
 & = & \lim_{h\rightarrow0}h^{-1}P(t\leq C<t+h\mid C\geq t,\overline{H}_{t},T>t)=\lambda_{C}\left(t\mid\overline{H}_{t},T>t\right),\label{eq:censoring}
\end{eqnarray}
denoted by $\lambda_{C}\left(t\mid\overline{H}_{t}\right)$ for shorthand. 

\end{assumption}

Assumption \ref{asp:NUC-1} states that $\lambda_{C}(t\mid F,T>t)$
depends only on the past treatment and covariate history until time
$t$, but not on the future variables and failure time. This assumption
holds if the set of historical covariates contains all prognostic
factors for the failure time that affect the lost to follow up at
time $t$. Under this assumption, the missing data due to censoring
are missing at random \citep{rubin1976inference}. In the presence
of censoring, redefine $V$ as the time to treatment discontinuation
or failure or censoring, whichever came first. We show in the supplementary
material that $\lambda_{V}(t\mid\overline{H}_{t})$ is equal to $\lambda_{V}(t\mid\overline{H}_{t},C\geq t)$
and therefore can be estimated conditional on $V\geq t$ with the
new definition of $V.$ From $\lambda_{C}\left(t\mid\overline{H}_{t}\right)$,
we define $K_{C}\left(t\mid\overline{H}_{t}\right)=\exp\left\{ -\int_{0}^{t}\lambda_{C}\left(u\mid\overline{H}_{u}\right)\de u\right\} ,$
which is the probability of the subject not being censored before
time $t$. For regularity, we also impose a positivity condition for
$K_{C}\left(t\mid\overline{H}_{t}\right)$.

\begin{assumption}[Positivity]\label{asp:positivity}There exists
a constant $\delta$ such that with probability one, $K_{C}\left(t\mid\overline{H}_{t}\right)\geq\delta>0$
for $t$ in the support of $T$.

\end{assumption}

Under Assumptions \ref{asump:UNC}\textendash \ref{asp:positivity},
$\psi^{*}$ is identifiable; see the supplementary material for proof.
Following \citet{rotnitzky2009analysis}, the main idea of inverse
probability of censoring weighting is to re-distribute the weights
for the censored subjects to the remaining ``similar'' uncensored
subjects. 

\begin{theorem}\label{Thm:ipcw}Under Assumptions \ref{asump:UNC}\textendash \ref{asp:positivity},
the unbiased estimating equation for $\psi^{*}$ is 
\begin{equation}
P_{n}\left\{ \frac{\Delta}{K_{C}\left(T\mid\overline{H}_{T}\right)}G(\psi;F)\right\} =0,\label{eq:IPCW}
\end{equation}
where $G(\psi;F)$ is defined in (\ref{eq:G}).

\end{theorem}

Theorem \ref{Thm:ipcw} assumes that $\lambda_{C}(t\mid\overline{H}_{t})$
is known. Similar to $\lambda_{V}(t\mid\overline{H}_{t})$, we posit
a proportional hazards model with time-dependent covariates: 
\begin{equation}
\lambda_{C}\left(t\mid\overline{H}_{t}\right)=\lambda_{C,0}(t)\exp\left\{ \gamma_{C}^{\T}g_{C}(t,\overline{H}_{t})\right\} ,\label{eq:ph-C}
\end{equation}
where $\lambda_{C,0}(t)$ is unknown and non-negative, $g_{C}(t,\overline{H}_{t})$
is a pre-specified function of $t$ and $\overline{H}_{t}$, and $\gamma_{C}$
is a vector of unknown parameters.

To summarize, the algorithm for developing an estimator of $\psi^{*}$
is as follows. 
\begin{description}
\item [{Step$\ $1.}] Using the data $(V_{i},\Gamma_{i},\overline{H}_{V_{i},i})$,
$i=1,\ldots,n$, fit a model for $\lambda_{V}\left(t\mid\overline{H}_{t}\right)=\lambda_{V,0}(t)\exp\left\{ \gamma_{V}^{\T}g_{V}(t,\overline{H}_{t})\right\} $.
To estimate $\gamma_{V}$, treat the treatment discontinuation as
``failure\textquotedblright{} and the failure event and censoring
as ``censored'' observations in the time-dependent proportional
hazards model. Once we have an estimate of $\gamma_{V}$, $\widehat{\gamma}_{V},$
we can estimate the cumulative baseline hazard, $\lambda_{V,0}(t)\de t$
using the Breslow estimator 
\[
\widehat{\lambda}_{V,0}(t)\de t=\frac{\sum_{i=1}^{n}\de N_{V,i}(t)}{\sum_{i=1}^{n}\exp\left\{ \widehat{\gamma}_{V}^{\T}g_{V}(t,\overline{H}_{t,i})\right\} Y_{V_{i}}(t)}.
\]
Then, obtain $\widehat{M}_{V}(t)=N_{V}(t)-\int_{0}^{t}\exp\left\{ \widehat{\gamma}_{V}^{\T}g_{V}(u,\overline{H}_{u})\right\} \widehat{\lambda}_{V,0}(u)Y_{V}(u)\de u$. 
\item [{Step$\ $2.}] Using the data $(X_{i},\Delta_{i},\overline{H}_{X_{i},i})$,
$i=1,\ldots,n$, derive the estimator of $\lambda_{C}\left(t\mid\overline{H}_{t}\right)=\lambda_{C,0}(t)\exp\left\{ \gamma_{C}^{\T}g_{C}(t,\overline{H}_{t})\right\} $,
and obtain the estimator of $K_{C}(T_{i}\mid\overline{H}_{T_{i}})$.
To estimate $\gamma_{C}$, treat censoring as ``failure\textquotedblright{}
and the failure event as ``censored'' observations in the time-dependent
proportional hazards model. Once we have an estimate of $\gamma_{C}$,
$\widehat{\gamma}_{C},$ we can estimate $\lambda_{C,0}(t)\de t$
using the Breslow estimator 
\[
\widehat{\lambda}_{C,0}(t)\de t=\frac{\sum_{i=1}^{n}\de N_{C,i}(t)}{\sum_{i=1}^{n}\exp\left\{ \widehat{\gamma}_{C}^{\T}g_{C}(t,\overline{H}_{t,i})\right\} Y_{C_{i}}(t)}
\]
where $N_{C}(t)=I(C\leq t,\text{\ensuremath{\Delta}=0})$ and $Y_{C}(t)=I(C\geq t)$
are the counting process and the at-risk process of observing censoring.
Then, we estimate $K_{C}\left(t\mid\overline{H}_{t}\right)$ by 
\[
\widehat{K}_{C}\left(t\mid\overline{H}_{t}\right)=\prod_{0\leq u\leq t}\left[1-\exp\left\{ \widehat{\gamma}_{C}^{\T}g_{C}(u,\overline{H}_{u})\right\} \widehat{\lambda}_{C,0}\left(u\right)\de u\right].
\]
\item [{Step$\ $3.}] We obtain the estimator $\widehat{\psi}$ of $\psi$
by solving 
\begin{equation}
P_{n}\left\{ \frac{\Delta}{\widehat{K}_{C}\left(T\mid\overline{H}_{T}\right)}\int c(\overline{H}_{u})\left[U(\psi)-E\left\{ U(\psi)\mid\overline{H}_{u},V\geq u;\widehat{\xi}\right\} \right]\de\widehat{M}_{V}(u)\right\} =0,\label{eq:IPCW2}
\end{equation}
where we estimate $E\left\{ U(\psi)\mid\overline{H}_{u},V\geq u;\xi\right\} $
by regressing $\widehat{K}_{C}\left(T\mid\overline{H}_{T}\right)^{-1}\Delta U(\psi)$
on $(X_{0},L_{u},u)$ restricted to subjects with $V\geq u$. The
estimating equation (\ref{eq:IPCW2}) is continuously differentiable
on $\psi$ and thus can be generally solved using a Newton-Raphson
procedure \citep{atkinson1989introduction}. For example, one can
use the function ``multiroot'' in R.
\end{description}
\begin{remark}It is worth discussing the connection between the proposed
estimator and the existing framework for discrete time points. If
the processes take observations at discrete times $\{t_{0},\ldots,t_{K}\}$,
then, for $t=t_{m}$, $\overline{H}_{t}=\{H_{t_{1}},\ldots,H_{t_{m}}\}$,
$\de N_{T}(t)$ is a binary treatment indicator, and $\int_{0}^{t}\lambda_{T}(u\mid\overline{H}_{u})Y_{T}(u)\de u$
becomes the propensity score $\pr\{\de N_{T}(t)=1\mid\overline{H}_{t}\}$.
As a result, (\ref{eq:IPCW2}) with $E\left\{ U(\psi)\mid\overline{H}_{u},V\geq u;\widehat{\xi}\right\} $
being zero simplifies to the existing estimating equation for $\psi^{*}$.
Importantly, (\ref{eq:IPCW2}), for the first time in the literature,
provides the semiparametric doubly robust estimator $\widehat{\psi}$
even for discrete time setting, in that $\widehat{\psi}$ is consistent
if either the model for the treatment process or the failure time
model is correctly specified, under correct model specifications for
the treatment effect mechanism and the censoring.

\end{remark}

\subsection{Asymptotic theory and variance estimation}

In this section we discuss the asymptotic properties of our proposed
estimator with technical details presented in the supplementary material.
To reflect the dependence of the estimating equation on the nuisance
models, denote (\ref{eq:IPCW2}) as $P_{n}\Phi(\psi,\widehat{\xi},\widehat{M}_{V},\widehat{K}_{C};F)=0$,
where $\Phi(\psi,\xi,M_{V},K_{C};F)=\{K_{C}\left(T\mid\overline{H}_{T}\right)\}^{-1}\Delta\int c(\overline{H}_{u})[U(\psi)-E\{U(\psi)\mid\overline{H}_{u},V\geq u;\xi\}]\de M_{V}(u)$.
Let the probability limits of $\widehat{\xi}$, $\widehat{M}_{V}$,
and $\widehat{K}_{C}$ be $\xi^{*}$, $M_{V}^{*}$, and $K_{C}^{*}$,
respectively. We impose standard regularity conditions for $Z$-estimators
\citep{van1996weak} as formulated by Assumptions \ref{asump:donsker}\textendash \ref{asump: IF}.
Roughly speaking, these conditions restrict the flexibility and convergence
rates of the nuisance estimators; e.g., we assume that $\Phi(\psi,\xi,M_{V},K_{C};F)$
and $\partial\Phi(\psi,\xi,M_{V},K_{C};F)/\partial\psi$ belong to
$P$-Donsker classes, and the regularity conditions ensure that
\begin{multline*}
E\left(\int c(\overline{H}_{u})\left[E\left\{ \left(\begin{array}{c}
U(\psi^{*})\\
\partial U(\psi^{*})/\partial\psi
\end{array}\right)\mid\overline{H}_{u},V\geq u;\widehat{\xi}\right\} \right.\right.\\
-\left.\left.E\left\{ \left(\begin{array}{c}
U(\psi^{*})\\
\partial U(\psi^{*})/\partial\psi
\end{array}\right)\mid\overline{H}_{u},V\geq u;\xi^{*}\right\} \right]\de\left\{ \widehat{M}_{V}(u)-M_{V}^{*}(u)\right\} \right)=o_{p}(n^{-1/2}).
\end{multline*}
Under Assumptions \ref{asp:positivity} and \ref{asump:donsker}\textendash \ref{asump: IF},
Theorem \ref{thm:s1} states that if $K_{C}$ is correctly specified,
and if either $E\{U(\psi)\mid\overline{H}_{u},V\geq u;\xi\}$ or $M_{V}$
is correctly specified, $\widehat{\psi}$ solving (\ref{eq:IPCW})
with the estimated nuisance models is still consistent and asymptotically
normal, with the influence function $\widetilde{\Phi}(\psi^{*},\xi^{*},M_{V}^{*},K_{C}^{*};F)$.

We can estimate the variance of $\widehat{\psi}$ either by the empirical
variance of the estimated influence function or by resampling. If
all nuisance models, $\xi,$ $M_{V}$, and $K_{C}$, are correctly
specified, we obtain an analytical expression for $\widetilde{\Phi}(\psi^{*},\xi^{*},M_{V}^{*},K_{C}^{*};F)$
as in (\ref{eq:tilde-J}). We can then estimate $\widetilde{\Phi}(\psi^{*},\xi^{*},M_{V}^{*},K_{C}^{*};F)$
by plugging in estimates of $\psi^{*}$, $\xi^{*}$, $M_{V}^{*}$,
$K_{C}^{*}$, and the required expectations, denoted by $\widehat{\Phi}(\widehat{\psi},\widehat{\xi},\widehat{M}_{V},\widehat{K}_{C};F)$.
Then, the estimated variance of $n^{1/2}(\widehat{\psi}-\psi^{*})$
is 
\begin{equation}
P_{n}\left\{ \widehat{\Phi}(\widehat{\psi},\widehat{\xi},\widehat{M}_{V},\widehat{K}_{C};F)\widehat{\Phi}(\widehat{\psi},\widehat{\xi},\widehat{M}_{V},\widehat{K}_{C};F)^{\T}\right\} .\label{eq:VE}
\end{equation}
However, when either $\xi$ or $M_{V}$ is correctly specified but
not both, characterizing $\widetilde{\Phi}(\psi^{*},\xi^{*},M_{V}^{*},K_{C}^{*};F)$
is difficult, and therefore approximating (\ref{eq:VE}) is no longer
feasible. To avoid the technical difficulty, we recommend estimating
the asymptotic variance with the resampling methods such as bootstrap
and Jackknife \citep{efron1979,efron1981jackknife}. In this case,
the resampling works because $\widehat{\psi}$ is regular and asymptotically
normal.

\section{Simulation study\label{sec:Simulation-Study}}

We evaluate the finite sample performance of the proposed estimator
on simulated data sets. We generate $U$ from Exp$(0.2)$ and generate
the covariate process $(X_{0},L_{t})$ had the treatment always been
withheld, where $X_{0}\sim$Bernoulli($0.55$). To generate $L_{t},$
we first generate a $1\times3$ row vector following a multivariate
normal distribution with mean equal to\textbf{ $(0.2U-4)$} and covariance
equal to $0.7^{|i-j|}$ for $i,j=1,2,3$. This vector represents the
values of $L_{t}$ at times $t_{1}=0$, $t_{2}=5$, and $t_{3}=10$.
We assume that the time-dependent variable remains constant between
measurements. We generate the time until treatment discontinuation,
$V_{1}$, according to a proportional hazards model $\lambda_{V}(t\mid X_{0},\overline{L}_{t})=0.15\exp(0.15X_{0}+0.15L_{t}).$
This generates the treatment process $A_{t}$; i.e., $A_{t}=1$ if
$t\leq V_{1}$ and $A_{t}=0$ if $t>V_{1}$. The observed time-dependent
covariate process is $L_{t}$ if $t\le V_{1}$ and $L_{t}+\log(t-V_{1})$
if $t>V_{1}$ to reflect that the covariate process is affected after
treatment discontinuation. Let the history of covariates and treatment
until time $t$ be $\overline{H}_{t}=(X_{0},\overline{L}_{t},\overline{A}_{t-})$.
We generate $T$ according to $U\sim\int_{0}^{T}\exp(\psi^{*}A_{u})\de u$
as follows. Let $T_{1}=U\exp(-\psi^{*})$. If $T_{1}<V_{1}$, $T=T_{1}$;
otherwise $T=U+V_{1}-V_{1}\exp(\psi^{*})$. Under the above data generating
mechanism, the potential failure time under $\overline{a}_{T}$ also
follows a Cox marginal structural model with the hazard rate at $u$,
$\lambda_{0}(u)\exp(\psi^{*}A_{u})$ \citep{young2010relation}. We
generate $C$ according to a proportional hazards model with $\lambda_{C}(t\mid X_{0},\overline{L}_{t},C\geq t)=0.025\exp(0.15X_{0}+0.15L_{t}).$
Let $X=\min(T,C)$. If $T<C$, $\Delta=1$; otherwise $\Delta=0$.
Finally, let $V=\min(V_{1},T,C)$ and $\Gamma$ be the indicator of
treatment discontinuation before the time to failure or censoring;
i.e., if $V=V_{1}$, $\Gamma=1$; otherwise $\Gamma=0$. The observed
data are $(X_{i},\Delta_{i},V_{i},\Gamma_{i},\overline{H}_{X_{i},i})$
for $i=1,\ldots,n$. We consider $\psi^{*}\in\{-0.5,0,0.5\}$. From
our data generating mechanism, $50\%-58\%$ observations are censored,
and $70\%-80\%$ treatment discontinuation times are observed before
the time to failure or censoring.

We consider the following estimators of $\psi^{*}$: (i) an naive
estimator $\widehat{\psi}_{\naive}$ by solving (\ref{eq:ee4}) with
$T$ in $U(\psi)=\int_{0}^{T}\exp(\psi^{*}A_{u})\de u$ replaced by
$X$; (ii) an inverse probability of weighting estimator of the Cox
marginal structural model $\widehat{\psi}_{\msm}$ \citep{yang2018modeling};
(iii) a simple inverse probability of censoring weighting estimator
$\widehat{\psi}_{\ipcw}$ by solving $P_{n}[\{\widehat{K}_{C}\left(T\mid\overline{H}_{T}\right)\}^{-1}\Delta\int c(\overline{H}_{u})U(\psi)\de M_{V}(u)]=0;$
and (iv) the proposed doubly robust estimator $\widehat{\psi}_{\dr}$
by solving (\ref{eq:IPCW2}) with $E\left\{ U(\psi)\mid\overline{H}_{u},V\geq u\right\} $
reducing to a tractable function $E\left\{ U(\psi)\mid\overline{H}_{0}\right\} $.
Note that $\widehat{\psi}_{\ipcw}$ is the special case of $\widehat{\psi}_{\dr}$
with $E\left\{ U(\psi)\mid\overline{H}_{u},V\geq u\right\} $ being
misspecified as zero. Moreover, to demonstrate the impact of data
discretization, we include the discrete-time g-estimator $\widehat{\psi}_{\disc}$
applied to the pre-processed data with the grid size $51$. We present
the details for $\widehat{\psi}_{\msm}$ and $\widehat{\psi}_{\disc}$
in the supplementary material. For estimators requiring the choice
of $c(\overline{H}_{u})$, we compare a simple choice $c(\overline{H}_{u})=A_{u-}$
and the optimal choice $c^{\opt}(\overline{H}_{u})$ in (\ref{eq:c-opt}),
where $E\{\partial\dot{U}_{u}(\psi)/\partial\psi\mid\overline{H}_{u},V=u\}=E(V-u\mid\overline{H}_{u},V\geq u)$.
We approximate $E(V-u\mid\overline{H}_{u},V\geq u)$ by the mean of
exponential distribution with the rate $\widehat{\lambda}_{V}(u)$
and assume that $\mathrm{var}\left\{ U(\psi)\mid\overline{H}_{u},V\geq u\right\} $
is a constant, which is common practice in the generalized estimating
equation literature. We approximate $E\left\{ U(\psi)\mid\overline{H}_{u},V\geq u\right\} $
by regressing $\widehat{K}_{C}\left(T\mid\overline{H}_{T}\right)^{-1}\Delta U(\psi)$
on $(X_{0},L_{0})$. To evaluate the double robustness, we consider
two specifications for the hazard of treatment discontinuation: (a)
the true proportional hazards model, and (b) a misspecified Kaplan-Meier
model \citep{Kaplan1958}. In calculating the censoring weights, we
specify the censoring model as the true proportional hazards model.
We assess the impact of misspecification of the censoring model in
the supplementary material. For standard errors, we consider the delete-a-group
Jackknife variance estimator with $500$ groups \citep{kott1998jackk2}.

\begin{table}
\caption{\label{tab:Results1}Simulation results: bias, standard deviation,
root mean squared error, and coverage rate of $95\%$ confidence intervals
for $\exp(\psi^{*})$ over $1,000$ simulated datasets: Scenario 1/2
the treatment discontinuation model is correctly specified/misspecified}

\centering{}\resizebox{\textwidth}{!}{

\begin{tabular}{cccccccccccc}
 &  &  & \multicolumn{3}{c}{$\psi{}^{*}=-0.5$} & \multicolumn{3}{c}{$\psi{}^{*}=0$} & \multicolumn{3}{c}{$\psi{}^{*}=0.5$}\tabularnewline
\multicolumn{3}{c}{} & Bias & S.E. & C.R. & Bias & S.E. & C.R. & Bias & S.E. & C.R.\tabularnewline
 & \multirow{2}{*}{$\widehat{\psi}_{\naive}$} & $c$ & 0.06 & 0.048 & 76.8 & 0.02 & 0.069 & 95.6 & -0.06 & 0.112 & 92.4\tabularnewline
 &  & $c^{\opt}$ & 0.05 & 0.043 & 78.4 & 0.02 & 0.063 & 95.0 & -0.05 & 0.107 & 91.8\tabularnewline
 & \multirow{2}{*}{$\widehat{\psi}_{\ipcw}$} & $c$ & -0.01 & 0.089 & 95.2 & -0.02 & 0.123 & 97.2 & -0.02 & 0.191 & 95.6\tabularnewline
Scenario 1 &  & $c^{\opt}$ & -0.01 & 0.070 & 96.4 & -0.02 & 0.095 & 97.0 & -0.02 & 0.148 & 95.6\tabularnewline
 & \multirow{2}{*}{$\widehat{\psi}_{\dr}$} & $c$ & 0.00 & 0.053 & 95.2 & -0.00 & 0.076 & 96.8 & -0.01 & 0.125 & 95.4\tabularnewline
 &  & $c^{\opt}$ & 0.00 & 0.049 & 95.4 & -0.00 & 0.071 & 96.0 & -0.00 & 0.118 & 94.8\tabularnewline
 & $\widehat{\psi}_{\msm}$ &  & -0.00 & 0.050 & 95.8 & 0.00 & 0.081 & 96.4 & 0.00 & 0.148 & 95.2\tabularnewline
 & $\widehat{\psi}_{\disc}$ &  & -0.37 & 0.041 & 0.0 & -0.61 & 0.055 & 0.0 & -1.01 & 0.092 & 0.6\tabularnewline
 & \multirow{2}{*}{$\widehat{\psi}_{\naive}$} & $c$ & 0.22 & 0.065 & 4.8 & 0.24 & 0.097 & 30.4 & 0.26 & 0.164 & 66.0\tabularnewline
 &  & $c^{\opt}$ & 0.22 & 0.066 & 5.4 & 0.24 & 0.097 & 31.8 & 0.26 & 0.163 & 68.2\tabularnewline
 & \multirow{2}{*}{$\widehat{\psi}_{\ipcw}$} & $c$ & 0.16 & 0.098 & 62.4 & 0.23 & 0.140 & 64.4 & 0.33 & 0.239 & 79.6\tabularnewline
Scenario 2 &  & $c^{\opt}$ & 0.16 & 0.098 & 62.2 & 0.23 & 0.140 & 65.8 & 0.33 & 0.234 & 79.4\tabularnewline
 & \multirow{2}{*}{$\widehat{\psi}_{\dr}$} & $c$ & 0.01 & 0.048 & 95.0 & 0.00 & 0.070 & 96.4 & 0.00 & 0.115 & 95.4\tabularnewline
 &  & $c^{\opt}$ & 0.01 & 0.048 & 95.4 & 0.00 & 0.070 & 96.6 & 0.00 & 0.115 & 95.2\tabularnewline
 & $\widehat{\psi}_{\msm}$ &  & 0.13 & 0.069 & 54.4 & -0.40 & 0.051 & 57.6 & 0.36 & 0.217 & 75.6\tabularnewline
 & $\widehat{\psi}_{\disc}$ &  & -0.25 & 0.035 & 0.0 & 0.22 & 0.118 & 0.0 & -0.72 & 0.092 & 1.0\tabularnewline
\end{tabular}}
\end{table}

Table \ref{tab:Results1} summarizes the simulation results with $n=1,000$.
The naive estimator $\widehat{\psi}_{\naive}$ is biased, and its
bias becomes larger as $|\psi^{*}|$ increases. In scenario 1 where
the treatment process model is correctly specified, $\widehat{\psi}_{\ipcw}$
, $\widehat{\psi}_{\dr}$, and $\widehat{\psi}_{\msm}$ show small
biases across all scenarios with different values $\psi^{*}$. Note
that $\widehat{\psi}_{\ipcw}$ is a special case of the proposed estimator
with $E\left\{ U(\psi)\mid\overline{H}_{u},V\geq u\right\} $ being
misspecified as zero. This demonstrates that the proposed estimator
is robust to misspecification of $E\left\{ U(\psi)\mid\overline{H}_{u},V\geq u\right\} $
given that the treatment process model is correctly specified. If
additionally $E\left\{ U(\psi)\mid\overline{H}_{u},V\geq u\right\} $
is well approximated, $\widehat{\psi}_{\dr}$ gains estimation efficiency
over $\widehat{\psi}_{\ipcw}$. Moreover, $\widehat{\psi}_{\dr}$
with $c^{\opt}$ are more efficient than that with $c$. Moreover,
in scenario 1, $\widehat{\psi}_{\dr}$ has smaller standard errors
than $\widehat{\psi}_{\msm}$. This is because $\widehat{\psi}_{\msm}$
involves weighting directly by the inverse of the propensity score,
whereas $\widehat{\psi}_{\dr}$ utilizes the propensity score not
in a form of inverse weights and therefore avoids the possibly large
variability due to weighting. In scenario 2 where the treatment process
model is misspecified, $\widehat{\psi}_{\ipcw}$ and $\widehat{\psi}_{\msm}$
show large biases; however, $\widehat{\psi}_{\dr}$ still has small
biases, confirming its double robustness. The Jackknife variance estimation
performs well for $\widehat{\psi}_{\dr}$ and produces coverage rates
close to the nominal coverage. Lastly, we note large biases in the
discrete-time g-estimator $\widehat{\psi}_{\disc}$, which illustrates
the consequence of data pre-processing for the subsequent analysis.

\section{Application to the GARFIELD data\label{sec:Application}}

We present an analysis for the Global Anticoagulant Registry in the
FIELD with Atrial Fibrillation (GARFIELD-AF) registry study, an observational
study of patients with newly diagnosed atrial fibrillation. See the
study website at http://www.garfieldregistry.org/ for details. Our
analysis includes $22,811$ patients, who were enrolled between April
2013 and August 2016 and received oral anticoagulant therapy for stroke
prevention. Our goal is to investigate the effect of discontinuation
of oral anticoagulant therapy in patients with atrial fibrillation.
The primary end point is the composite clinical outcome including
death, non-haemorrhagic stroke, systemic embolism, and myocardial
infarction. We define a patient as permanently discontinuing if treatment
was stopped for at least $7$ days and never re-started afterwards.
In our study, $9.5\%$ of patients discontinued oral anticoagulant
therapy over a median follow-up of $710$ days with an interquartile
range $(487,731)$ days; 43.8\% of discontinuations were within the
first $4$ months of the start of treatment. Among those who discontinued
treatment, $512$ patients stopped the treatment for more than $7$
days and went back on treatment. This is called switching. We censor
the switches at the time of restarting treatment. This censoring mechanism
is not likely to be completely at random, because patients with poor
prognosis may be more likely to switch. We assume a dependent censoring
mechanism and use inverse probability of censoring weighting.

To answer the clinical question of interest, we consider the structural
failure time model $U(\psi^{*})=\int_{0}^{T}\exp(\psi^{*}A_{u})\de u$.
Under this model, if a patient had been on treatment continuously,
\textbf{$T=U(\psi^{*})\exp(-\psi^{*})$}, so $U(\psi^{*})\{\exp(-\psi^{*})-1\}$
is the time gained/reduced while on treatment. We focus on estimating
the multiplicative factor $\exp(\psi^{*})$. Table \ref{tab:Results}
reports the results from the naive estimator and the proposed doubly
robust estimator as described in $\mathsection$ \ref{sec:Simulation-Study}.
We describe the details for the nuisance models in the supplementary
material. Although the effect sizes may be a little different between
the naive analysis and the proposed analysis, qualitatively they all
suggest that treatment is beneficial for prolonging the time to clinical
events, and therefore treatment discontinuation is harmful. If a patient
had been on treatment continuously versus if the patient had never
taken treatment, the time to clinical outcomes would have been $\exp(-\widehat{\psi})=1/0.64=1.56$
times longer. Importantly, the proposed analysis is designed to address
the well-formulated question for investigating the effect of treatment
discontinuation.

\begin{table}
\caption{\label{tab:Results}Results of the effect of oral anticoagulant therapy
on the composite outcome: $\exp(\psi^{*})$ is the causal estimand }

\centering{}%
\begin{tabular}{ccccc}
 & Est  & S.E.  & C.I.  & p-value \tabularnewline
Naive method  & 0.68  & 0.176  & (0.34,1.03)  & 0.07\tabularnewline
Proposed method  & 0.64  & 0.179  & (0.29,0.99)  & 0.04\tabularnewline
\end{tabular}
\end{table}

\section{Discussion\label{sec:Discussions}}

The proposed framework of structural failure time model can be used
to adjust for time-varying confounding and selection bias with irregularly
spaced observations under three assumptions of no unmeasured confounders,
ignorability of censoring, and positivity. As discussed previously,
Assumptions 1 and 2 hold if all variables that are related to both
treatment discontinuation and outcome and that are related to both
censoring and outcome are measured. Although essential, they are not
verifiable based on the observed data but rely on subject matter experts
to assess their plausibility. The future work will investigate the
sensitivity to these assumptions using the methods in \citet{yang2017sensitivity}.
Assumption 3 states that all subjects have nonzero probabilities of
staying on study before the failure time. This assumption requires
the absence of predictors that are deterministic in relation to censoring
and outcome. Practitioners should carefully examine the question at
hand to eliminate deterministic violations of positivity.

Our framework can also be extended in the following directions. First,
the proposed doubly robust estimator with respect to model specifications
for the treatment process and the baseline failure time; however,
it still relies on a correct specification of the censoring mechanism.
If the censoring model is misspecified, the proposed estimator may
be biased; see the additional simulation results in the supplementary
material. It would be interesting to construct an improved estimator
that is multiply robust in the sense that such an estimator is consistent
in the union of the three models \citep{molina2017multiple}. Second,
it is critical to derive test procedures for evaluating the goodness-of-fit
of the treatment effect model. The key insight is that we have more
unbiased estimating equations than the model parameters. In future
work, we will derive tests based on over-identification restrictions
tests \citep{yang2015gof} for evaluating a treatment effect model.

\section*{Acknowledgment}

We benefited from the comments from two reviewers and Anastasio A.
Tsiatis. Dr. Yang is partially supported by ORAU, NSF DMS 1811245,
and NCI P01 CA142538.

\section*{Supplementary Material}

Supplementary material available at \textit{Biometrika} online includes
proofs, technical details and additional simulation. R package is
available at https://github.com/shuyang1987/contTimeCausal.

\bibliographystyle{dcu}
\bibliography{C:/Dropbox/bib/ci}

\newpage{}

\global\long\def\theequation{S\arabic{equation}}%
 \setcounter{equation}{0}

\global\long\def\thelemma{S\arabic{lemma}}%
 \setcounter{lemma}{0}

\global\long\def\theexample{S\arabic{example}}%
 \setcounter{equation}{0}

\global\long\def\thesection{S\arabic{section}}%
 \setcounter{section}{0}

\global\long\def\thetheorem{S\arabic{theorem}}%
 \setcounter{equation}{0}

\global\long\def\thecondition{S\arabic{condition}}%
 \setcounter{equation}{0}

\global\long\def\thetable{S\arabic{table}}%
 \setcounter{equation}{0}

\global\long\def\theremark{S\arabic{remark}}%
 \setcounter{equation}{0}

\global\long\def\thestep{S\arabic{step}}%
 \setcounter{equation}{0}

\global\long\def\theassumption{S\arabic{assumption}}%
 \setcounter{assumption}{0}

\global\long\def\theproof{S\arabic{proof}}%
 \setcounter{equation}{0}

\global\long\def\theproposition{S{proposition}}%
 \setcounter{equation}{0} 

\setcounter{page}{1}

\section*{Supplementary Material}

\section{A lemma}

We provide a lemma for the martingale process, which is useful in
our derivation later.

Consider the Hilbert space $\mathcal{H}$ of all $p$-dimensional,
mean-zero finite variance measurable functions of $F$, $h(F)$, equipped
with the covariance inner product $<h_{1},h_{2}>=E\left\{ h_{1}(F)^{\T}h_{2}(F)\right\} $
and the norm $||h||=\left[E\left\{ h(F)^{\T}h(F)\right\} \right]{}^{1/2}<\infty$.

\begin{lemma}\label{lemma:1}Under Assumption \ref{asump:UNC}, $M_{V}(t)$
is a martingale with respect to the filtration $\sigma\{\overline{H}_{t},U(\psi^{*})\}$.
By Proposition II.4.1 in \citet{andersen2012statistical}, $M_{V}(t)$
has an unique compensator $<M_{V}(t)>=\int_{0}^{t}\lambda_{V}(u\mid\overline{H}_{u})Y_{V}(u)\de u$.
If $g_{1}(\cdot)$ and $g_{2}(\cdot)$ are bounded $\sigma\{\overline{H}_{t},U(\psi^{*})\}$-predictable
processes, then 
\[
<\int_{0}^{t}g_{1}(u)\de M_{V}(u),\int_{0}^{t}g_{2}(u)\de M_{V}(u)>
\]
exists, and 
\begin{equation}
<\int_{0}^{t}g_{1}(u)\de M_{V}(u),\int_{0}^{t}g_{2}(u)\de M_{V}(u)>=\int_{0}^{t}g_{1}(u)g_{2}(u)\lambda_{V}(u\mid\overline{H}_{u})Y_{V}(u)\de u.\label{eq:martingale1}
\end{equation}

\end{lemma}

\section{Proof of (\ref{eq:UNC2})}

To show (\ref{eq:UNC2}), it suffices to show that $\lambda_{V}\{t\mid\overline{H}_{t},U(\psi^{*})\}=\lambda_{V}(t\mid\overline{H}_{t},U)$.
We obtain 
\begin{eqnarray*}
 &  & \lambda_{V}\left(t\mid\overline{H}_{t},U\right)=\lim_{h\rightarrow0}h^{-1}P\left(t\leq V<t+h,\Gamma=1\mid V\geq t,\overline{H}_{t},U\right)\\
 & = & \lim_{h\rightarrow0}h^{-1}\frac{P\left(U\mid t\leq V<t+h,\Gamma=1,\overline{H}_{t}\right)P(t\leq V<t+h,\Gamma=1\mid V\geq t,\overline{H}_{t})}{P\left(U\mid V\geq t,\Gamma=1,\overline{H}_{t}\right)}\\
 & = & \lim_{h\rightarrow0}h^{-1}\frac{P\left\{ U(\psi^{*})\mid t\leq V<t+h,\Gamma=1,\overline{H}_{t}\right\} P(t\leq V<t+h,\Gamma=1\mid V\geq t,\overline{H}_{t})}{P\left\{ U(\psi^{*})\mid V\geq t,\Gamma=1,\overline{H}_{t}\right\} }\\
 & = & \lim_{h\rightarrow0}h^{-1}P\left\{ t\leq V<t+h,\Gamma=1\mid V\geq t,\overline{H}_{t},U(\psi^{*})\right\} \\
 & = & \lambda_{V}\{t\mid\overline{H}_{t},U(\psi^{*})\},
\end{eqnarray*}
where the second equality follows by the Bayes rule, and the third
equality follows by Model (\ref{eq:trt model}) which entails that
the distributions of $(U,\overline{H}_{t})$ and $\{U(\psi^{*}),\overline{H}_{t}\}$
are the same.

\section{Identification of $\psi\in\R^{p}$ under Assumption \ref{asump:UNC}
\label{sec:Identification}}

Under Assumption \ref{asump:UNC}, $M_{V}(t)=N_{V}(t)-\int_{0}^{t}\lambda_{V}(u\mid\overline{H}_{u})Y_{V}(u)\de u$
is a martingale with respect to the filtration $\sigma\{\overline{H}_{t},U(\psi^{*})\}$.
Then, for any $c(\overline{H}_{t})\in\mathbb{\R}^{p}$ and $t>0$,
\begin{equation}
E\left\{ c(\overline{H}_{t})U(\psi^{*})\de M_{V}(t)\right\} =0.\label{eq:eeforidentification}
\end{equation}
Suppose that (\ref{eq:eeforidentification}) holds for $\psi^{*1}$
and $\psi^{*2}$; i.e., for any $c(\overline{H}_{t})\in\mathbb{\R}^{p}$
and $t>0$, $E[c(\overline{H}_{t})\{U(\psi^{*1})-U(\psi^{*2})\}\de M_{V}(t)]=0.$
To reflect the dependence of $U(\psi^{*1})-U(\psi^{*2})$ on $\left(\overline{A}_{T},\overline{L}_{T}\right)$,
denote $\varphi\left(\overline{A}_{T},\overline{L}_{T}\right)=U(\psi^{*1})-U(\psi^{*2})=\int_{0}^{T}\exp[\{\psi_{1}^{*1}+\psi_{2}^{*1\T}g(L_{u})\}A_{u}]\de u-\int_{0}^{T}\exp[\{\psi_{1}^{*2}+\psi_{2}^{*2\T}g(L_{u})\}A_{u}]\de u$.
Then, for any $c(\overline{H}_{t})\in\mathbb{\R}^{p}$ and $t>0$,
we have $E\left\{ c(\overline{H}_{t})\varphi\left(\overline{A}_{T},\overline{L}_{T}\right)\de M_{V}(t)\right\} =0.$
This implies that $\varphi\left(\overline{A}_{T},\overline{L}_{T}\right)$
is independent of $M_{V}(t)$ conditional on $(\overline{H}_{t},V>t)$
for all $\overline{H}_{t}$ and $t>0$. Therefore, $\varphi\left(\overline{A}_{T},\overline{L}_{T}\right)$
must not depend on $\overline{A}_{T}$, and therefore $\psi^{*1}$
must equal $\psi^{*2}$. Consequently, $\psi^{*}$ is uniquely identified
from (\ref{eq:eeforidentification}).

\section{Proof of Theorem \ref{Thm:Ortho}}

To motivate the concept of the nuisance tangent space for a semiparametric
model, we first consider a parametric model $f(F;\psi,\theta)$, where
$\psi$ is a $p$-dimensional parameter of interest, and $\theta$
is an $q$-dimensional nuisance parameter. The score vectors of $\psi$
and $\theta$ are $S_{\psi}(F)=\partial\log f(F;\psi,\theta^{*})/\partial\psi$
and $S_{\theta}(F)=\partial\log f(F;\psi^{*},\theta)/\partial\theta$,
respectively, both evaluated at the true value $(\psi^{*},\theta^{*})$.
For this parametric model, the nuisance tangent space $\Lambda$ is
the linear space in $\mathcal{H}$ spanned by the nuisance score vector
$S_{\theta}(F)$. In a semiparametric model, the nuisance parameter
$\theta$ may be infinite-dimensional. The nuisance tangent space
$\Lambda$ is defined as the mean squared closure of the nuisance
tangent spaces under any parametric submodel. An important fact is
that the orthogonal complement of the nuisance tangent space $\Lambda^{\bot}$
contains the influence functions for regular asymptotically linear
estimators of $\psi$.

First, we characterize the semiparametric likelihood function based
on a single observable $F$. Because the transformation of $F$ to
$\{U(\psi^{*}),\overline{H}_{T}\}$ is one-to-one, the likelihood
function based on $F$ becomes
\begin{equation}
f_{F}\left(T,\overline{H}_{T}\right)=\left\{ \frac{\partial U(\psi^{*})}{\partial T}\right\} f_{\{U(\psi^{*}),\overline{H}_{T}\}}\{U(\psi^{*}),\overline{H}_{T}\},\label{eq:slik}
\end{equation}
where $\partial U(\psi^{*})/\partial T=\exp\left[A_{T}\{\psi_{1}^{*}+\psi_{2}^{*\T}g(L_{T})\}\right]$.
Let $v_{0}=0<v_{1}<\cdots<v_{M}$ be the observed times to treatment
discontinuation among the $n$ subjects. We further express (\ref{eq:slik})
as 
\begin{eqnarray}
f_{F}\left(T,\overline{H}_{T};\psi^{*},\theta\right) & = & \left\{ \frac{\partial U(\psi^{*})}{\partial T}\right\} f\left\{ U(\psi^{*});\theta_{1}\right\} \prod_{k=1}^{M}f\left\{ L_{v_{k}}\mid\overline{H}_{v_{k}-1},U(\psi^{*}),T>v_{k};\theta_{2}\right\} \nonumber \\
 &  & \times\prod_{v=v_{1}}^{v_{M}}f\left\{ A_{v_{k}}\mid\overline{H}_{v_{k}-1},U(\psi^{*}),T>v_{k};\theta_{3}\right\} \nonumber \\
 & = & \left\{ \frac{\partial U(\psi^{*})}{\partial T}\right\} f\left\{ U(\psi^{*});\theta_{1}\right\} \prod_{k=1}^{M}f\left\{ L_{v_{k}}\mid\overline{H}_{v_{k}-1},U(\psi^{*}),T>v_{k};\theta_{2}\right\} \nonumber \\
 &  & \times\prod_{v=v_{1}}^{v_{M}}f\left(A_{v_{k}}\mid\overline{H}_{v_{k}-1},T>v_{k};\theta_{3}\right),\label{eq:slik2}
\end{eqnarray}
where the second equality follows from Assumption \ref{asump:UNC}
and (\ref{eq:UNC2}), $f\left\{ U(\psi^{*})\right\} $, $f\left\{ L_{v_{k}}\mid\overline{H}_{v_{k}-1},U(\psi^{*}),T>v_{k}\right\} ,$
and $f\left(A_{v_{k}}\mid\overline{H}_{v_{k}-1},T>v_{k}\right)$ are
completely unspecified, and $\theta=(\theta_{1},\theta_{2},\theta_{3})$
is a vector of infinite-dimensional nuisance parameters.

Let $\Lambda_{k}$ be the nuisance tangent space for $\theta_{k},$
for $k=1,2,3$. We now characterize $\Lambda_{k}$. 

For the nuisance parameter $\theta_{1}$, $f\left\{ U(\psi^{*});\theta_{1}\right\} $
is a nonparametric model indexed by $\theta_{1}$, i.e., $f\left\{ U(\psi^{*});\theta_{1}\right\} $
is a non-negative function and satisfies $\int f\left(v;\theta_{1}\right)\de v=1$.
Following Section 4.4 of \citet{tsiatis2007semiparametric}, the tangent
space regarding $\theta_{1}$ is the set of all vector $s\left\{ U(\psi^{*})\right\} \in\mathbb{\R}^{p}$
with $E\left[s\left\{ U(\psi^{*})\right\} \right]=0$. Thus, the tangent
space of $\theta_{1}$ is 
\[
\Lambda_{1}=\left\{ s\left\{ U(\psi^{*})\right\} \in\mathbb{\R}^{p}:E\left[s\left\{ U(\psi^{*})\right\} \right]=0\right\} .
\]

For the nuisance parameter $\theta_{2}$, $\prod_{k=1}^{M}f\left\{ L_{v_{k}}\mid\overline{H}_{v_{k}-1},U(\psi^{*}),T>v_{k};\theta_{2}\right\} $
is a nonparametric model indexed by $\theta_{2}$. To obtain the nuisance
tangent space of $\theta_{2}$, following the same derivation as for
$\theta_{1}$, the score function of $\theta_{2}$ is of the form
$\sum_{k=1}^{M}S\left\{ L_{v_{k}},\overline{H}_{v_{k}-1},U(\psi^{*})\right\} $,
where $E[S\left\{ L_{v_{k}},\overline{H}_{v_{k}-1},U(\psi^{*})\right\} \mid\overline{H}_{v_{k}-1},U(\psi^{*}),T>v_{k}]=0$.
Thus, the tangent space of $\theta_{2}$ is
\[
\Lambda_{2}=\sum_{k=1}^{M}\left\{ S\left\{ L_{v_{k}},\overline{H}_{v_{k}-1},U(\psi^{*})\right\} \in\R^{p}:E\left[S\left\{ L_{v_{k}},\overline{H}_{v_{k}-1},U(\psi^{*})\right\} \mid\overline{H}_{v_{k}-1},U(\psi^{*}),T>v_{k}\right]=0\right\} .
\]

For the nuisance parameter $\theta_{3}$, $\prod_{k=1}^{M}f\left(A_{v_{k}}\mid\overline{L}_{v_{k}-1},\overline{A}_{v_{k}-1},T>v_{k};\theta_{3}\right)$
can be equivalently expressed as the likelihood based on the data
$(V,\Gamma,\overline{H}_{V})$ and the hazard function $\lambda_{V}(t\mid\overline{H}_{t})$:
\begin{multline*}
f_{(V,\Gamma,\overline{H}_{V})}(V,\Gamma,\overline{H}_{V})=\lambda_{V}(V\mid\overline{H}_{V})^{\Gamma}\exp\left\{ -\int_{0}^{V}\lambda_{V}(u\mid\overline{H}_{u})\de u\right\} \\
\times\left\{ f_{T\mid\overline{H}_{T}}(V\mid\overline{H}_{V})\right\} ^{1-\Gamma}\left\{ \int_{V}^{\infty}f_{T\mid\overline{H}_{T}}(u\mid\overline{H}_{u})\de u\right\} ^{\Gamma}.
\end{multline*}
Following \citet{tsiatis2007semiparametric}, the tangent space of
$\theta_{3}$ is 
\[
\Lambda_{3}=\left\{ \int h_{u}(\overline{H}_{u})\de M_{V}(u):\ h_{u}(\overline{H}_{u})\in\mathbb{\R}^{p}\right\} .
\]

Moreover, it is easy to show that $\Lambda_{1}$, $\Lambda_{2}$ and
$\Lambda_{3}$ are mutually orthogonal subspaces. Then, $\Lambda=\Lambda_{1}\oplus\Lambda_{2}\oplus\Lambda_{3}$,
where $\oplus$ denotes a direct sum.

Now, let 
\[
\Lambda_{3}^{*}=\left\{ \int h_{u}\{U(\psi^{*}),\overline{H}_{u}\}\de M_{V}(u):\ h_{u}\{U(\psi^{*}),\overline{H}_{u}\}\in\R^{p}\right\} .
\]
Because the tangent space $\Lambda_{1}\oplus\Lambda_{2}\oplus\Lambda_{3}^{*}$
is that for a nonparametric model; i.e., a model that allows for all
densities of $F$, and because the tangent space for a nonparametric
model is the entire Hilbert space, we obtain $\mathcal{H}=\Lambda_{1}\oplus\Lambda_{2}\oplus\Lambda_{3}^{*}.$
Because $\Lambda=\Lambda_{1}\oplus\Lambda_{2}\oplus\Lambda_{3}$,
this implies that $\Lambda_{3}\subset\Lambda_{3}^{*}$. Also, the
orthogonal complement $\Lambda^{\bot}$ must be orthogonal to $\Lambda_{1}\oplus\Lambda_{2}$,
so $\Lambda^{\bot}$ must belong to $\Lambda_{3}^{*}$ and be orthogonal
to $\Lambda_{3}$. This means that $\Lambda^{\bot}$ consists of all
elements of $\Lambda_{3}^{*}$ that are orthogonal to $\Lambda_{3}$.

To characterize $\Lambda^{\bot}$, for any $\int h_{u}\{U(\psi^{*}),\overline{H}_{u}\}\de M_{V}(u)\in\Lambda_{3}^{*}$,
we obtain its projection onto $\Lambda_{3}^{\bot}$. To find the projection,
we derive $h_{u}^{*}(\overline{H}_{u})$ so that 
\[
\left[\int h_{u}\{U(\psi^{*}),\overline{H}_{u}\}\de M_{V}(u)-\int h_{u}^{*}(\overline{H}_{u})\de M_{V}(u)\right]\in\Lambda_{3}^{\bot}.
\]
Therefore, we have 
\begin{equation}
E\left(\int\left[h_{u}\{U(\psi^{*}),\overline{H}_{u}\}-h_{u}^{*}(\overline{H}_{u})\right]\de M_{V}(u)\times\int h_{u}(\overline{H}_{u})\de M_{V}(u)\right)=0,\label{eq:eq1}
\end{equation}
for any $h_{u}(\overline{H}_{u})$. By Lemma \ref{lemma:1}, (\ref{eq:eq1})
becomes 
\begin{eqnarray*}
 &  & E\left(\int\left[h_{u}\{U(\psi^{*}),\overline{H}_{u}\}-h_{u}^{*}(\overline{H}_{u})\right]h_{u}(\overline{H}_{u})\lambda_{V}(u\mid\overline{H}_{u})Y_{V}(u)\de u\right)\\
 & = & E\left(\int E\left(\left[h_{u}\{U(\psi^{*}),\overline{H}_{u}\}-h_{u}^{*}(\overline{H}_{u})\right]Y_{V}(u)\mid\overline{H}_{u}\right)h_{u}(\overline{H}_{u})\lambda_{V}(u\mid\overline{H}_{u})\de u\right)=0
\end{eqnarray*}
for any $h_{u}(\overline{H}_{u})$. Because $h_{u}(\overline{H}_{u})$
is arbitrary, we must have 
\begin{equation}
E\left(\left[h_{u}\{U(\psi^{*}),\overline{H}_{u}\}-h_{u}^{*}(\overline{H}_{u})\right]Y_{V}(u)\mid\overline{H}_{u}\right)=0.\label{eq:eq2}
\end{equation}
Solving (\ref{eq:eq2}) for $h_{u}^{*}(\overline{H}_{u})$, we obtain
\[
E\left[h_{u}\{U(\psi^{*}),\overline{H}_{u}\}Y_{V}(u)\mid\overline{H}_{u}\right]=h_{u}^{*}(\overline{H}_{u})E\left\{ Y_{V}(u)\mid\overline{H}_{u}\right\} ,
\]
or
\[
h_{u}^{*}(\overline{H}_{u})=\frac{E\left[h_{u}\{U(\psi^{*}),\overline{H}_{u}\}Y_{V}(u)\mid\overline{H}_{u}\right]}{E\left\{ Y_{V}(u)\mid\overline{H}_{u}\right\} }=E\left[h_{u}\{U(\psi^{*}),\overline{H}_{u}\}\mid\overline{H}_{u},V\geq u\right].
\]
Therefore, the space orthogonal to the nuisance tangent space is given
by
\[
\Lambda^{\bot}=\left\{ \int\left(h_{u}\{U(\psi^{*}),\overline{H}_{u}\}-E\left[h_{u}\{U(\psi^{*}),\overline{H}_{u}\}\mid\overline{H}_{u},V\geq u\right]\right)\de M_{V}(u):\vphantom{\sum_{i}}h_{u}\{U(\psi^{*}),\overline{H}_{u}\}\in\mathbb{\R}^{p}\right\} .
\]

\section{The optimal form $c^{\opt}(\overline{H}_{u})$}

We obtain the optimal form of $c(\overline{H}_{u})$ by projecting
the score function $S_{\psi}(F)$ onto
\[
\Lambda_{0}^{\bot}=\left\{ G(\psi^{*};F,c)=\int_{0}^{\infty}c(\overline{H}_{u})\left[U(\psi^{*})-E\left\{ U(\psi^{*})\mid\overline{H}_{u},V\geq u\right\} \right]\de M_{V}(u):c(\overline{H}_{u})\in\mathbb{\R}^{p}\right\} .
\]

We first characterize the projection of any $B(F)\in\mathcal{H}$
onto $\Lambda_{0}^{\bot}$. For ease of notation, we may suppress
the dependence of $F$ of random variables if there is no ambiguity.

\begin{theorem}[Projection]\label{projection} For any $B=B(F)\in\mathcal{H}$,
the projection of $B$ onto $\Lambda_{0}^{\bot}$ is
\begin{eqnarray}
 &  & \prod\left(B\mid\Lambda_{0}^{\bot}\right)=\int\left[E\left\{ B\dot{U}_{u}(\psi^{*})\mid\overline{H}_{u},V=u\right\} -E\left\{ B\dot{U}_{u}(\psi^{*})\mid\overline{H}_{u},V\geq u\right\} \right]\nonumber \\
 &  & \times\left[\mathrm{var}\left\{ U(\psi^{*})\mid\overline{H}_{u},V\geq u\right\} \right]^{-1}\left[U(\psi^{*})-E\left\{ U(\psi^{*})\mid\overline{H}_{u},V\geq u\right\} \right]\de M_{V}(u),\label{eq:projection}
\end{eqnarray}
where $\dot{U}_{u}(\psi)=U(\psi)-E\{U(\psi)\mid\overline{H}_{u},V\geq u\}$.

\end{theorem}

\textit{Proof. }Let $G(F)$ be the quantity in the right hand side
of (\ref{eq:projection}). To show that $\prod\left(B\mid\Lambda_{0}^{\bot}\right)=G(F)$,
we must show that $B-G\in\Lambda_{0}$. Toward that end, we show that
for any $\widetilde{G}(F)\in\Lambda_{0}^{\bot}$, $(B-G)\indep\widetilde{G}$.
Specifically, we need to show that for any $\widetilde{G}(F)=\int_{0}^{\infty}\widetilde{c}(\overline{H}_{u})\left[U(\psi^{*})-E\left\{ U(\psi^{*})\mid\overline{H}_{u},V\geq u\right\} \right]\de M_{V}(u)$,
$E\left\{ (B-G)\widetilde{G}\right\} =0$. We now verify that $E\left(B\widetilde{G}\right)=E\left(G\widetilde{G}\right)$
by the following calculation.

Firstly, we obtain
\begin{eqnarray}
 &  & E\left(G\widetilde{G}\right)=E\left(<G,\widetilde{G}>\right)\nonumber \\
 & = & E\int\widetilde{c}(\overline{H}_{u})\left[E\left\{ BU(\psi^{*})\mid\overline{H}_{u},V=u\right\} -E\left\{ BU(\psi^{*})\mid\overline{H}_{u},V\geq u\right\} \right]\label{eq:right}\\
 &  & \times\left[\var\left\{ U(\psi^{*})\mid\overline{H}_{u},V\geq u\right\} \right]^{-1}\left[U(\psi^{*})-E\left\{ U(\psi^{*})\mid\overline{H}_{u},V\geq u\right\} \right]^{2}\lambda_{V}(u\mid\overline{H}_{u})Y_{V}(u)\de u\nonumber \\
 & = & E\int\widetilde{c}(\overline{H}_{u})\left[E\left\{ BU(\psi^{*})\mid\overline{H}_{u},V=u\right\} -E\left\{ BU(\psi^{*})\mid\overline{H}_{u},V\geq u\right\} \right]\lambda_{V}(u\mid\overline{H}_{u})Y_{V}(u)\de u.\nonumber 
\end{eqnarray}

Secondly, we obtain
\begin{eqnarray}
 &  & E\left(B\widetilde{G}\right)=E\int\widetilde{c}(\overline{H}_{u})B\left[U(\psi^{*})-E\left\{ U(\psi^{*})\mid\overline{H}_{u},T\geq u\right\} \right]\de M_{V}(u)\nonumber \\
 & = & E\int\widetilde{c}(\overline{H}_{u})B\dot{U}_{u}(\psi^{*})\de N_{V}(u)-E\int_{0}^{\infty}\widetilde{c}(\bar{V}_{u})B\dot{U}_{u}(\psi^{*})\lambda_{V}(u\mid\overline{H}_{u})Y_{V}(u)\de u\label{eq:left}\\
 & = & E\int\widetilde{c}(\overline{H}_{u})\left[E\left\{ B\dot{U}_{u}(\psi^{*})\mid\overline{H}_{u},V=u\right\} -E\left\{ B\dot{U}_{u}(\psi^{*})\mid\overline{H}_{u},V\geq u\right\} \right]\lambda_{V}(u\mid\overline{H}_{u})Y_{V}(u)\de u,\nonumber 
\end{eqnarray}
where the last equality follows because
\begin{eqnarray*}
E\int\widetilde{c}(\overline{H}_{u})B\dot{U}_{u}(\psi^{*})\de N_{V}(u) & = & E\int\widetilde{c}(\overline{H}_{u})E\left\{ B\dot{U}_{u}(\psi^{*})\de N_{V}(u)\mid\overline{H}_{u},V\geq u\right\} \\
 & = & E\int\widetilde{c}(\overline{H}_{u})E\left\{ B\dot{U}_{u}(\psi^{*})I\left(u\leq V\leq u+\de u,\Gamma=1\right)\mid\overline{H}_{u},V\geq u\right\} \\
 & = & E\int\widetilde{c}(\overline{H}_{u})E\left\{ B\dot{U}_{u}(\psi^{*})\mid\overline{H}_{u},V=u\right\} \lambda_{V}(u\mid\overline{H}_{u})Y_{V}(u)\de u,
\end{eqnarray*}
and
\[
E\int\widetilde{c}(\overline{H}_{u})B\dot{U}_{u}(\psi^{*})\lambda_{V}(u\mid\overline{H}_{u})Y_{V}(u)\de u=E\int\widetilde{c}(\bar{V}_{u})E\left\{ B\dot{U}_{u}(\psi^{*})\mid\overline{H}_{u},V\geq u\right\} \lambda_{V}(u\mid\overline{H}_{u})Y_{V}(u)\de u.
\]
Therefore, by (\ref{eq:right}) and (\ref{eq:left}), $E\left(B\widetilde{G}\right)=E\left(G\widetilde{G}\right)$
for any $\widetilde{G}\in\Lambda_{0}^{\bot}$, proving (\ref{eq:projection}).

\begin{theorem}\label{Thm: efficient score}The optimal form of $c(\overline{H}_{u})$
is (\ref{eq:c-opt}) in the sense that with this form the solution
to (\ref{eq:ee4}) gives the most precise estimator of $\psi^{*}$
among all the solutions to (\ref{eq:ee4}).

\end{theorem}

\textit{Proof.} We write $G(\psi^{*};F,c)$ to emphasize its dependence
on $c(\overline{H}_{u})$. We derive the optimal form of $c(\overline{H}_{u})$
by deriving the most efficient $G(\psi^{*};F,c)$ in $\Lambda_{0}^{\bot}$,
which is $G(\psi^{*};F,c^{\opt})=\prod\left(S_{\psi}\mid\Lambda_{0}^{\bot}\right)$.

By Theorem \ref{projection}, we have
\begin{multline}
G(\psi^{*};F,c^{\opt})=\int\left[E\left\{ S_{\psi}\dot{U}_{u}(\psi^{*})\mid\overline{H}_{u},V=u\right\} -E\left\{ S_{\psi}\dot{U}_{u}(\psi^{*})\mid\overline{H}_{u},V\geq u\right\} \right]\\
\times\left[\mathrm{var}\left\{ U(\psi^{*})\mid\overline{H}_{u},V\geq u\right\} \right]^{-1}\left[U(\psi^{*})-E\left\{ U(\psi^{*})\mid\overline{H}_{u},V\geq u\right\} \right]\de M_{V}(u).\label{eq:G-opt0}
\end{multline}
Because $E\{\dot{U}_{u}(\psi)\mid\overline{H}_{u},V\geq u\}=0$, taking
the derivative of $\psi$ at both sides and using the generalized
information equality, we have $E\{S_{\psi}\dot{U}_{u}(\psi)\mid\overline{H}_{u},V\geq u\}+E\{\partial\dot{U}_{u}(\psi)/\partial\psi\mid\overline{H}_{u},V\geq u\}=0$,
or equivalently $E\{S_{\psi}\dot{U}_{u}(\psi)\mid\overline{H}_{u},V\geq u\}=-E\{\partial\dot{U}_{u}(\psi)/\partial\psi\mid\overline{H}_{u},V\geq u\}$.
Similarly, because $E\{\dot{U}_{u}(\psi)\mid\overline{H}_{u},V=u\}=0$,
we have $E\{S_{\psi}\dot{U}_{u}(\psi)\mid\overline{H}_{u},V=u\}+E\{\partial\dot{U}_{u}(\psi)/\partial\psi\mid\overline{H}_{u},V=u\}=0$,
or equivalently $E\{S_{\psi}\dot{U}_{u}(\psi)\mid\overline{H}_{u},V=u\}=-E\{\partial\dot{U}_{u}(\psi)/\partial\psi\mid\overline{H}_{u},V=u\}$.
Continuing (\ref{eq:G-opt0}),
\begin{eqnarray}
G(\psi^{*};F,c^{\opt}) & = & -\int_{0}^{\infty}\left[E\left\{ \partial\dot{U}_{u}(\psi^{*})/\partial\psi\mid\overline{H}_{u},V=u\right\} -E\left\{ \partial\dot{U}_{u}(\psi^{*})/\partial\psi\mid\overline{H}_{u},V\geq u\right\} \right]\nonumber \\
 &  & \times\left[\mathrm{var}\left\{ U(\psi^{*})\mid\overline{H}_{u},V\geq u\right\} \right]^{-1}\left[U(\psi^{*})-E\left\{ U(\psi^{*})\mid\overline{H}_{u},V\geq u\right\} \right]\de M_{V}(u)\nonumber \\
 & = & -\int_{0}^{\infty}E\left\{ \partial\dot{U}_{u}(\psi^{*})/\partial\psi\mid\overline{H}_{u},V=u\right\} \left[\mathrm{var}\left\{ U(\psi^{*})\mid\overline{H}_{u},V\geq u\right\} \right]^{-1}\nonumber \\
 &  & \times\left[U(\psi^{*})-E\left\{ U(\psi^{*})\mid\overline{H}_{u},V\geq u\right\} \right]\de M_{V}(u).\label{eq:G-opt}
\end{eqnarray}
Therefore, by (\ref{eq:G-opt}), ignoring the negative sign, $c^{\opt}(\overline{H}_{u})$
is given by (\ref{eq:c-opt}).

\section{Proof of Theorem \ref{Thm:2-dr}}

We show that $E\{G(\psi^{*};F,c)\}=0$ in two cases.

First, if $\lambda_{V}(t\mid\overline{H}_{t})$ is correctly specified,
under Assumption \ref{asump:UNC}, $M_{V}(t)$ is a martingale with
respect to the filtration $\sigma\{\overline{H}_{t},U(\psi^{*})\}$.
Because $c(\overline{H}_{u})\left[U(\psi^{*})-E\left\{ U(\psi^{*})\mid\overline{H}_{u},V\geq u\right\} \right]$
is a $\sigma\{\overline{H}_{t},U(\psi^{*})\}$-predictable process,
$\int_{0}^{t}c(\overline{H}_{u})\left[U(\psi^{*})-E\left\{ U(\psi^{*})\mid\overline{H}_{u},V\geq u\right\} \right]\de M_{V}(u)$
is a martingale for $t\geq0$. Therefore, $E\{G(\psi^{*};F,c)\}=0$.

Second, if $E\left\{ U(\psi^{*})\mid\overline{H}_{u},V\geq u\right\} $
is correctly specified but $\lambda_{V}(t\mid\overline{H}_{t})$ is
not necessarily correctly specified, let $\lambda_{V}^{*}(t\mid\overline{H}_{t})$
be the probability limit of the possibly misspecified model. We obtain
\begin{eqnarray}
 &  & E\int c(\overline{H}_{u})\left[U(\psi^{*})-E\left\{ U(\psi^{*})\mid\overline{H}_{u},V\geq u;\xi^{*}\right\} \right]\left\{ \de N_{V}(u)-\lambda_{V}^{*}(u\mid\overline{H}_{u})Y_{V}(u)\de u\right\} \nonumber \\
 & = & E\int c(\overline{H}_{u})\left[U(\psi^{*})-E\left\{ U(\psi^{*})\mid\overline{H}_{u},V\geq u;\xi^{*}\right\} \right]\left\{ \de N_{V}(u)-\lambda_{V}(u\mid\overline{H}_{u})Y_{V}(u)\de u\right\} \nonumber \\
 &  & +E\int c(\overline{H}_{u})\left[U(\psi^{*})-E\left\{ U(\psi^{*})\mid\overline{H}_{u},V\geq u;\xi^{*}\right\} \right]\left\{ \lambda_{V}(u\mid\overline{H}_{u})-\lambda_{V}^{*}(u\mid\overline{H}_{u})\right\} Y_{V}(u)\de u\nonumber \\
 & = & 0+E\int c(\overline{H}_{u})E\left(\left[U(\psi^{*})-E\left\{ U(\psi^{*})\mid\overline{H}_{u},V\geq u;\xi^{*}\right\} \right]\mid\overline{H}_{u},V\geq u\right)\label{eq:eq3}\\
 &  & \times\left\{ \lambda_{V}(u\mid\overline{H}_{u})-\lambda_{V}^{*}(u\mid\overline{H}_{u})\right\} Y_{V}(u)\de u\nonumber \\
 & = & 0+E\int c(\overline{H}_{u})\times0\times\left\{ \lambda_{V}(u\mid\overline{H}_{u})-\lambda_{V}^{*}(u\mid\overline{H}_{u})\right\} Y_{V}(u)\de u\label{eq:eq4}\\
 & = & 0,\nonumber 
\end{eqnarray}
where zero in (\ref{eq:eq3}) follows because $\de M_{V}(u)=\de N_{V}(u)-\lambda_{V}(u\mid\overline{H}_{u})\de u$
is a martingale with respect to the filtration $\sigma\{\overline{H}_{t},U(\psi^{*})\}$,
and zero in (\ref{eq:eq4}) follows because $E\left\{ U(\psi^{*})\mid\overline{H}_{u},V\geq u\right\} $
is correctly specified and therefore, $E\left\{ U(\psi^{*})\mid\overline{H}_{u},V\geq u;\xi^{*}\right\} =E\left\{ U(\psi^{*})\mid\overline{H}_{u},V\geq u\right\} $.

\section{Proof that $\widetilde{U}(\psi^{*})$ and $\Delta(\psi^{*})$ are
computable}

If $T\leq C$, because $U(\psi^{*})$ and $C(\psi^{*})$ are observable,
$\widetilde{U}(\psi^{*})$ and $\Delta(\psi^{*})$ are computable.
If $C<T$, $U(\psi^{*})$ is not computable; however, in this case,
we shall show that $C(\psi^{*})<U(\psi^{*})$ corresponding to $\widetilde{U}(\psi^{*})=C(\psi^{*})$
and $\Delta(\psi^{*})=0$, which are computable. Toward this end,
by definition of $C(\psi^{*})$, we show that when $C<T$, it is always
the case that $C(\psi^{*})\leq U(\psi^{*})$. If $\psi^{*}\geq0$,
$C(\psi^{*})=C\leq T\leq\int_{0}^{T}\exp(\psi^{*}A_{u})\de u=U(\psi^{*})$.
If $\psi^{*}<0$, $C(\psi^{*})=C\exp(\psi^{*})\leq T\exp(\psi^{*})=\int_{0}^{T}\exp(\psi^{*})\de u\leq\int_{0}^{T}\exp(\psi^{*}A_{u})\de u=U(\psi^{*})$.
This completes the proof.

\section{Proof of $\lambda_{V}(t\mid\overline{H}_{t})=\lambda_{V}(t\mid\overline{H}_{t},C\protect\geq t)$}

First, by Assumption \ref{asump:UNC}, we obtain
\begin{eqnarray*}
P(C\ge t\mid t\leq V<t+h,\Gamma=1,\overline{H}_{t}) & = & \exp\left\{ \int_{0}^{t}-\lambda_{C}(u\mid t\leq V<t+h,\Gamma=1,\overline{H}_{t})\de u\right\} \\
 & = & \exp\left\{ \int_{0}^{t}-\lambda_{C}(u\mid\overline{H}_{u})\de u\right\} ,
\end{eqnarray*}
and similarly, we obtain
\begin{eqnarray*}
P(C\ge t\mid V\ge t,\Gamma=1,\overline{H}_{t}) & = & \exp\left\{ \int_{0}^{t}-\lambda_{C}(u\mid V\ge t,\Gamma=1,\overline{H}_{t})\de u\right\} \\
 & = & \exp\left\{ \int_{0}^{t}-\lambda_{C}(u\mid\overline{H}_{u})\de u\right\} .
\end{eqnarray*}
Consequently, $P(C\ge t\mid t\leq V<t+h,\Gamma=1,\overline{H}_{t})=P(C\ge t\mid V\ge t,\Gamma=1,\overline{H}_{t})$.

Now, by the Bayes rule, we express
\begin{multline*}
\lambda_{V}(t\mid\overline{H}_{t},C\geq t)=\lim_{h\rightarrow0}h^{-1}P(t\leq V<t+h,\Gamma=1\mid V\geq t,\overline{H}_{t},C\ge t)
\end{multline*}
\begin{eqnarray*}
 & = & \lim_{h\rightarrow0}h^{-1}\frac{P(t\leq V<t+h,\Gamma=1\mid V\geq t,\overline{H}_{t})P(C\ge t\mid t\leq V<t+h,\Gamma=1,\overline{H}_{t})}{P(C\ge t\mid V\ge t,\overline{H}_{t})}\\
 & = & \lim_{h\rightarrow0}h^{-1}P(t\leq V<t+h,\Gamma=1\mid V\geq t,\overline{H}_{t})=\lambda_{V}(t\mid\overline{H}_{t}).
\end{eqnarray*}

\section{Identification of $\psi\in\R^{p}$ under Assumptions \ref{asump:UNC},\textendash \ref{asp:positivity}}

Under Assumptions \ref{asp:NUC-1}, and \ref{asp:positivity}, for
any $c(\overline{H}_{t})\in\mathbb{\R}^{p}$ and $t>0$, 
\begin{equation}
E\left\{ \frac{\Delta}{K_{C}\left(T\mid\overline{H}_{T}\right)}c(\overline{H}_{t})U(\psi^{*})\de M_{V}(t)\right\} =E\left\{ c(\overline{H}_{t})U(\psi^{*})\de M_{V}(t)\right\} =0.\label{eq:eeforidentification-1}
\end{equation}
Because under Assumption \ref{asump:UNC}, $\psi^{*}$ is uniquely
identified from (\ref{eq:eeforidentification}). Therefore, under
Assumptions \ref{asump:UNC}\textendash \ref{asp:positivity}, $\psi^{*}$
is uniquely identified from (\ref{eq:eeforidentification-1}).

\section{Proof of Theorem \ref{Thm:ipcw} \label{sec:Proof-of-ipcw}}

To show (\ref{eq:IPCW}) is an unbiased estimating equation, it suffices
to show that
\[
E\left\{ \frac{\Delta}{K_{C}\left(T\mid\overline{H}_{T}\right)}G(\psi^{*};F)\right\} =0.
\]
Toward that end, by the iterative expectation, we have
\begin{eqnarray*}
E\left\{ \frac{\Delta}{K_{C}\left(T\mid\overline{H}_{T}\right)}G(\psi^{*};F)\right\}  & = & E\left[E\left\{ \frac{\Delta}{K_{C}\left(T\mid\overline{H}_{T}\right)}G(\psi^{*};F)\mid F\right\} \right]\\
 & = & E\left\{ \frac{E(\Delta\mid F)}{K_{C}\left(T\mid\overline{H}_{T}\right)}G(\psi^{*};F)\right\} \\
 & = & E\left\{ 1\times G(\psi^{*};F)\right\} =0,
\end{eqnarray*}
where the third equality follows by the dependent censoring mechanism
specified in (\ref{eq:censoring}).

\section{The asymptotic properties of the proposed estimator}

To establish the asymptotic properties of the proposed estimator,
we first introduce additional notation.

Recall the nuisance models (i) $E\{U(\psi^{*})\mid\overline{H}_{u},V\geq u;\xi\}$
indexed by $\xi$; (ii) the proportional hazards model for the treatment
process (\ref{eq:ph-V}), indexed by $M_{V}$; and (iii) the proportional
hazards model for the censoring process (\ref{eq:ph-C}), indexed
by $K_{C}$. $\widehat{\xi}$, $\widehat{M}_{V}$, and $\widehat{K}_{C}$
are the estimates of $\xi$, $M_{V}$, and $K_{C}$ under the specified
parametric and semiparametric models. The probability limits of $\widehat{\xi}$,
$\widehat{M}_{V}$, and $\widehat{K}_{C}$ are $\xi^{*}$, $M_{V}^{*}$,
and $K_{C}^{*}$. If the failure time model is correctly specified,
$E\{U(\psi^{*})\mid\overline{H}_{u},V\geq u;\xi^{*}\}=E\{U(\psi^{*})\mid\overline{H}_{u},V\geq u\}$;
if the model for the treatment process is correctly specified, $M_{V}^{*}=M_{V}$;
and if the model for the censoring process is correctly specified,
$K_{C}^{*}=K_{C}$. 

To reflect that the estimating function depends on the nuisance parameters,
we define
\begin{eqnarray*}
\Phi(\psi,\xi,M_{V},K_{C};F) & = & \frac{\Delta G(\psi,\xi,M_{V};F)}{K_{C}\left(T\mid\overline{H}_{T}\right)},\\
G(\psi,\xi,M_{V};F) & = & \int c(\overline{H}_{u})\left[U(\psi)-E\left\{ U(\psi)\mid\overline{H}_{u},V\geq u;\xi\right\} \right]\de M_{V}(t).
\end{eqnarray*}
Let $P$ denote the true data generating distribution, and for any
$f(F)$, let $P\{f(F)\}=\int f(x)\de P(x)$. We define
\begin{eqnarray*}
J_{1}(\xi) & = & P\left\{ \Phi(\psi^{*},\xi,M_{V}^{*},K_{C}^{*};F)\right\} ,\\
J_{2}(M_{V}) & = & P\left\{ \Phi(\psi^{*},\xi^{*},M_{V},K_{C}^{*};F)\right\} ,\\
J_{3}(K_{C}) & = & P\left\{ \Phi(\psi^{*},\xi^{*},M_{V}^{*},K_{C};F)\right\} ,
\end{eqnarray*}
and
\[
J(\xi,M_{V},K_{C})=P\left\{ \Phi(\psi^{*},\xi,M_{V},K_{C};F)\right\} .
\]

We now assume the regularity conditions, which are standard in the
empirical process literature \citep{van1996weak}. See also \citet{yang2015gof}
for the application of the empirical process to derive a goodness-of-fit
test for the structural nested mean models.

\begin{assumption}\label{asump:donsker} With probability going to
one, $\Phi(\psi,\xi,M_{V},K_{C};F)$ and $\partial\Phi(\psi,\xi,M_{V},K_{C};F)/\partial\psi$
are $P$-Donsker classes.

\end{assumption}

\begin{assumption}\label{asump convg}For $(\xi^{*},M_{V}^{*},K_{C}^{*})$
with either $\xi^{*}$ being the true parameter such that $E\{U(\psi^{*})\mid\overline{H}_{u},V\geq u;\xi^{*}\}=E\{U(\psi^{*})\mid\overline{H}_{u},V\geq u\}$
or $M_{V}^{*}=M_{V}$, and $K_{C}^{*}=K_{C}$,
\[
P\left\{ ||\Phi(\psi^{*},\widehat{\xi},\widehat{M}_{V},\widehat{K}_{C};F)-\Phi(\psi^{*},\xi^{*},M_{V}^{*},K_{C}^{*};F)||\right\} \rightarrow0
\]
and
\[
P\left\{ ||\frac{\partial}{\partial\psi}\Phi(\widehat{\psi},\widehat{\xi},\widehat{M}_{V},\widehat{K}_{C};F)-\frac{\partial}{\partial\psi}\Phi(\psi^{*},\xi^{*},M_{V}^{*},K_{C}^{*};F)||\right\} \rightarrow0
\]
in probability.

\end{assumption}

\begin{assumption}\label{assump:inv} $A(\psi^{*},\xi^{*},M_{V}^{*},K_{C}^{*})=P\left\{ \partial\Phi(\psi^{*},\xi^{*},M_{V}^{*},K_{C}^{*};F)/\partial\psi\right\} $
is invertible.

\end{assumption}

\begin{assumption}\label{asump: IF}Assume that
\begin{eqnarray*}
J(\widehat{\xi},\widehat{M}_{V},\widehat{K}_{C})-J(\xi^{*},M_{V}^{*},K_{C}^{*}) & = & J_{1}(\widehat{\xi})-J_{1}(\xi^{*})+J_{2}(\widehat{M}_{V})-J_{2}(M_{V}^{*})\\
 &  & +J_{3}(\widehat{K}_{C})-J_{3}(K_{C}^{*})+o_{p}(n^{-1/2}),
\end{eqnarray*}
and that $J_{1}(\widehat{\xi})$, $J_{2}(\widehat{M}_{V})$, and $J_{3}(\widehat{K}_{C})$
are regular asymptotically linear with influence function $\Phi_{1}(\psi^{*},\xi^{*},M_{V}^{*},K_{C}^{*};F)$,
$\Phi_{2}(\psi^{*},\xi^{*},M_{V}^{*},K_{C}^{*};F)$, and $\Phi_{3}(\psi^{*},\xi^{*},M_{V}^{*},K_{C}^{*};F)$,
respectively.

\end{assumption}

Assumption \ref{asump:donsker} is an empirical process condition.
This assumption is technical and depends on the submodel chosen models
for the unknown parameters. Assuming the positivity condition for
the censoring process, this assumption can typically be considered
as a regularity condition.

Assumption \ref{asump convg} basically states that $\widehat{\xi}$,
$\widehat{M}_{V}$, and $\widehat{K}_{C}$ are consistent for $\xi^{*}$,
$M_{V}^{*}$, and $K_{C}$ and requires tha{\small{}t}
\begin{multline*}
E\left\{ \int c(\overline{H}_{u})\left[E\left\{ U(\psi^{*})\mid\overline{H}_{u},V\geq u;\widehat{\xi}\right\} \right.\right.\\
-\left.\left.E\left\{ U(\psi^{*})\mid\overline{H}_{u},V\geq u;\xi^{*}\right\} \right]\left\{ \widehat{\lambda}_{V}(u)-\lambda_{V}^{*}(u)\right\} \de u\right.=o_{p}(n^{-1/2}),
\end{multline*}
an{\small{}d}
\begin{multline*}
E\left\{ \int c(\overline{H}_{u})\left[E\left\{ \frac{\partial U(\psi^{*})}{\partial\psi}\mid\overline{H}_{u},V\geq u;\widehat{\xi}\right\} \right.\right.\\
-\left.\left.E\left\{ \frac{\partial U(\psi^{*})}{\partial\psi}\mid\overline{H}_{u},V\geq u;\xi^{*}\right\} \right]\left\{ \widehat{\lambda}_{V}(u)-\lambda_{V}^{*}(u)\right\} \de u\right\} =o_{p}(n^{-1/2}).
\end{multline*}

Because smooth functionals of parametric or semiparametric maximum
likelihood estimators for a given model are efficient under regularity
conditions, Assumption \ref{asump: IF} holds under regularity conditions
if $\widehat{\xi}$ and $\widehat{M}_{V}$ are the parametric and
semiparametric maximum likelihood estimators of $\xi^{*}$ and $M_{V}^{*}$
under the specified models.

We present the asymptotic properties of the proposed estimator $\widehat{\psi}$
solving equation (\ref{eq:IPCW}), denoted by $P_{n}\left\{ \Phi(\psi,\widehat{\xi},\widehat{M}_{V},\widehat{K}_{C};F)\right\} =0$.

\begin{theorem}\label{thm:s1}Under Assumptions \ref{asp:positivity}
and \ref{asump:donsker}\textendash \ref{asump: IF}, $n^{1/2}\left(\widehat{\psi}-\psi^{*}\right)$
is consistent and asymptotically linear with the influence function
$\widetilde{\Phi}(\psi^{*},\xi^{*},M_{V}^{*},K_{C}^{*};F)=\left\{ A(\psi^{*},\xi^{*},M_{V}^{*},K_{C}^{*})\right\} ^{-1}\widetilde{B}(\psi^{*},\xi^{*},M_{V}^{*},K_{C}^{*};F)$,
and
\begin{eqnarray}
\widetilde{B}(\psi^{*},\xi^{*},M_{V}^{*},K_{C}^{*};F) & = & \Phi(\psi^{*},\xi^{*},K_{V}^{*},K_{C}^{*};F)+\Phi_{1}(\psi^{*},\xi^{*},K_{V}^{*},K_{C}^{*};F)\nonumber \\
 &  & +\Phi_{2}(\psi^{*},\xi^{*},K_{V}^{*},K_{C}^{*};F)+\Phi_{3}(\psi^{*},\xi^{*},K_{V}^{*},K_{C}^{*};F).\label{eq:influence fctn}
\end{eqnarray}
Moreover, if the nuisance models including the models for the censoring
process and the treatment process and the outcome model are correctly
specified, (\ref{eq:influence fctn}) becomes
\begin{eqnarray}
 &  & \widetilde{B}(\psi^{*},\xi^{*},K_{V},K_{C};F)\nonumber \\
 & = & \Phi(\psi^{*},\xi^{*},K_{V},K_{C};F)-\prod\left\{ \Phi(\psi^{*},\xi^{*},K_{V},K_{C};F)\mid\widetilde{\Lambda}\right\} \nonumber \\
 & = & \Phi(\psi^{*},\xi^{*},K_{V},K_{C};F)-E\left\{ \Phi(\psi^{*},\xi^{*},K_{V},K_{C};F)S_{\gamma_{V}}^{\T}\right\} E\left(S_{\gamma_{V}}S_{\gamma_{V}}^{\T}\right)^{-1}S_{\gamma_{V}}\nonumber \\
 &  & -E\left\{ \Phi(\psi^{*},\xi^{*},K_{V},K_{C};F)S_{\gamma_{C}}^{\T}\right\} E\left(S_{\gamma_{C}}S_{\gamma_{C}}^{\T}\right)^{-1}S_{\gamma_{C}}\nonumber \\
 &  & +\int\frac{E\left[G(\psi^{*},\xi^{*},K_{V};F)\exp\left\{ \gamma_{C}^{\T}g_{C}(u,\overline{H}_{u})\right\} \Delta/K_{C}(T\mid\overline{H}_{T})\right]}{E\left[\exp\left\{ \gamma_{C}^{\T}g_{C}(u,\overline{H}_{u})\right\} Y_{C}(u)\right]}\de M_{C}(u)\nonumber \\
 &  & +\int\frac{E\left[G(\psi^{*},\xi^{*},K_{V};F)\exp\left\{ \gamma_{V}^{\T}g_{V}(u,\overline{H}_{u})\right\} \Delta/K_{C}(T\mid\overline{H}_{T})\right]}{E\left[\exp\left\{ \gamma_{V}^{\T}g_{V}(u,\overline{H}_{u})\right\} Y_{V}(u)\right]}\de M_{V}(u).\label{eq:tilde-J}
\end{eqnarray}

\end{theorem}

\begin{proof}

We assume that the model for the censoring process is correctly specified,
either the outcome model or the model for the treatment process is
correctly specified.

Taylor expansion of $P_{n}\left\{ \Phi(\widehat{\psi},\widehat{\xi},\widehat{M}_{V},\widehat{K}_{C};F)\right\} =0$
around $\psi^{*}$ leads to
\[
0=P_{n}\left\{ \Phi(\widehat{\psi},\widehat{\xi},\widehat{M}_{V},\widehat{K}_{C};F)\right\} =P_{n}\left\{ \Phi(\psi^{*},\widehat{\xi},\widehat{M}_{V},\widehat{K}_{C};F)\right\} +P_{n}\left\{ \frac{\partial\Phi(\widetilde{\psi},\widehat{\xi},\widehat{M}_{V},\widehat{K}_{C};F)}{\partial\psi^{\T}}\right\} (\widehat{\psi}-\psi^{*}),
\]
where $\widetilde{\psi}$ is on the line segment between $\widehat{\psi}$
and $\psi^{*}$.

Under Assumptions \ref{asump:donsker} and \ref{asump convg},
\[
(P_{n}-P)\left\{ \frac{\partial\Phi(\widetilde{\psi},\widehat{\xi},\widehat{M}_{V},\widehat{K}_{C};F)}{\partial\psi^{\T}}\right\} =(P_{n}-P)\left\{ \frac{\partial\Phi(\psi^{*},\xi^{*},M_{V}^{*},K_{C}^{*};F)}{\partial\psi^{\T}}\right\} =o_{p}(n^{-1/2}),
\]
and therefore,
\begin{eqnarray*}
P_{n}\left\{ \frac{\partial\Phi(\widetilde{\psi},\widehat{\xi},\widehat{M}_{V},\widehat{K}_{C};F)}{\partial\psi^{\T}}\right\}  & = & P\left\{ \frac{\partial\Phi(\widetilde{\psi},\widehat{\xi},\widehat{M}_{V},\widehat{K}_{C};F)}{\partial\psi^{\T}}\right\} +o_{p}(n^{-1/2})\\
 & = & A(\psi^{*},\xi^{*},M_{V}^{*},K_{C}^{*})+o_{p}(n^{-1/2}).
\end{eqnarray*}
We then have
\begin{equation}
n^{1/2}(\widehat{\psi}-\psi^{*})=\left\{ A(\psi^{*},\xi^{*},M_{V}^{*},K_{C}^{*})\right\} ^{-1}n^{1/2}P_{n}\left\{ \Phi(\psi^{*},\widehat{\xi},\widehat{M}_{V},\widehat{K}_{C};F)\right\} +o_{p}(1).\label{eq:1}
\end{equation}
To evaluate (\ref{eq:1}) further,
\begin{multline}
P_{n}\Phi(\psi^{*},\widehat{\xi},\widehat{M}_{V},\widehat{K}_{C};F)=(P_{n}-P)\Phi(\psi^{*},\widehat{\xi},\widehat{M}_{V},\widehat{K}_{C};F)\\
+P\left\{ \Phi(\psi^{*},\widehat{\xi},\widehat{M}_{V},\widehat{K}_{C};F)-\Phi(\psi^{*},\xi^{*},M_{V}^{*},K_{C}^{*};F)\right\} +P\Phi(\psi^{*},\xi^{*},M_{V}^{*},K_{C}^{*};F).\label{eq:2}
\end{multline}
Based on the double robustness, the third term becomes
\begin{equation}
P\Phi(\psi^{*},\xi^{*},M_{V}^{*},K_{C}^{*};F)=0.\label{eq:2-1}
\end{equation}
By Assumptions \ref{asump:donsker} and \ref{asump convg}, the first
term becomes
\begin{eqnarray}
(P_{n}-P)\Phi(\psi^{*},\widehat{\xi},\widehat{M}_{V},\widehat{K}_{C};F) & = & (P_{n}-P)\Phi(\psi^{*},\xi^{*},M_{V}^{*},K_{C}^{*};F)+o_{p}(n^{-1/2})\nonumber \\
 & = & P_{n}\Phi(\psi^{*},\xi^{*},M_{V}^{*},K_{C}^{*};F)+o_{p}(n^{-1/2}).\label{eq:2-2}
\end{eqnarray}
By Assumption \ref{asump: IF}, the second term becomes
\begin{eqnarray}
 &  & P\left\{ \Phi(\psi^{*},\widehat{\xi},\widehat{M}_{V},\widehat{K}_{C};F)-\Phi(\psi^{*},\xi^{*},M_{V}^{*},K_{C}^{*};F)\right\} \nonumber \\
 & = & J(\widehat{\xi},\widehat{M}_{V},\widehat{K}_{C})-J(\xi^{*},M_{V}^{*},K_{C}^{*})+o_{p}(n^{-1/2})\nonumber \\
 & = & J_{1}(\widehat{\xi})-J_{1}(\xi^{1})+J_{2}(\widehat{M}_{V})-J_{2}(M_{V}^{1})+J_{3}(\widehat{K}_{C})-J_{3}(M_{C}^{1})+o_{p}(n^{-1/2})\nonumber \\
 & = & P_{n}\left\{ \Phi_{1}(\psi^{*},\xi^{*},M_{V}^{*},K_{C}^{*};F)+\Phi_{2}(\psi^{*},\xi^{*},M_{V}^{*},K_{C}^{*};F)+\Phi_{3}(\psi^{*},\xi^{*},M_{V}^{*},K_{C}^{*};F)\right\} .\label{eq:2-3}
\end{eqnarray}

Combining (\ref{eq:2-2})\textendash (\ref{eq:2-1}) with (\ref{eq:2}),
\[
P_{n}\Phi(\psi^{*},\widehat{\xi},\widehat{M}_{V},\widehat{K}_{C};F)=P_{n}\{\widetilde{\Phi}(\psi^{*},\xi^{*},M_{V}^{*},K_{C}^{*};F)\},
\]
where
\begin{eqnarray*}
\widetilde{B}(\psi^{*},\xi^{*},M_{V}^{*},K_{C}^{*};F) & = & \Phi(\psi^{*},\xi^{*},M_{V}^{*},K_{C}^{*};F)+\Phi_{1}(\psi^{*},\xi^{*},M_{V}^{*},K_{C}^{*};F)\\
 &  & +\Phi_{2}(\psi^{*},\xi^{*},M_{V}^{*},K_{C}^{*};F)+\Phi_{3}(\psi^{*},\xi^{*},M_{V}^{*},K_{C}^{*};F).
\end{eqnarray*}

Therefore, $\widehat{\psi}-\psi^{*}$ has the influence function
\[
\widetilde{\Phi}(\psi^{*},\xi^{*},M_{V}^{*},K_{C}^{*};F)=\left\{ A(\psi^{*},\xi^{*},M_{V}^{*},K_{C}^{*})\right\} ^{-1}\widetilde{B}(\psi^{*},\xi^{*},M_{V}^{*},K_{C}^{*};F).
\]
As a result,
\begin{equation}
n^{1/2}(\widehat{\psi}-\psi^{*})=n^{1/2}P_{n}\widetilde{\Phi}(\psi^{*},\xi^{*},K_{V}^{*},K_{C}^{*};F)+o_{p}(1).\label{eq:(2.3)}
\end{equation}
Based on (\ref{eq:(2.3)}),
\[
n^{1/2}(\widehat{\psi}-\psi^{*})\rightarrow\N\left(0,\Omega\right),
\]
as $n\rightarrow\infty$, where $\Omega=E\left\{ \widetilde{\Phi}(\psi^{*},\xi^{*},M_{V}^{*},K_{C}^{*};F)\widetilde{\Phi}(\psi^{*},\xi^{*},M_{V}^{*},K_{C}^{*};F)^{\T}\right\} $.

For the special case where both nuisance models are correctly specified,
we characterize $\widetilde{B}(\psi^{*},\xi^{*},K_{V}^{*},K_{C}^{*};F)$.
In this case, $E\{U(\psi^{*})\mid\overline{H}_{u},V\geq u;\xi^{*}\}=E\{U(\psi^{*})\mid\overline{H}_{u},V\geq u\}$,
$M_{V}^{*}=M_{V},$ and $K_{C}^{*}=K_{C}$. Define the score functions:
$S_{\xi}=S_{\xi}\{U(\psi^{*}),\overline{H}_{u},V\geq u\}$,
\[
S_{\gamma_{V}}=\int\left\{ g_{V}(u,\overline{H}_{u})-\frac{E\left[g_{V}(u,\overline{H}_{u})\exp\left\{ \gamma_{V}^{\T}g_{V}(u,\overline{H}_{u})\right\} Y_{V}(u)\right]}{E\left[\exp\left\{ \gamma_{V}^{\T}g_{V}(u,\overline{H}_{u})\right\} Y_{V}(u)\right]}\right\} \de M_{V}(u),
\]
and
\[
S_{\gamma_{C}}=\int\left\{ g_{C}(u,\overline{H}_{u})-\frac{E\left[g_{C}(u,\overline{H}_{u})\exp\left\{ \gamma_{C}^{\T}g_{C}(u,\overline{H}_{u})\right\} Y_{C}(u)\right]}{E\left[\exp\left\{ \gamma_{C}^{\T}g_{C}(u,\overline{H}_{u})\right\} Y_{C}(u)\right]}\right\} \de M_{C}(u).
\]
The tangent space for $\xi$ is $\widetilde{\Lambda}_{1}=\{S_{\xi}\in\R^{p}:E(S_{\xi}\mid\overline{H}_{u},V\geq u)=0\}$.
Following \citet{tsiatis2007semiparametric}, the nuisance tangent
space for the proportional hazards model (\ref{eq:ph-V}) is
\[
\widetilde{\Lambda}_{2}=\left\{ S_{\gamma_{V}}+\int h(u)\de M_{V}(u):\ h(u)\in\mathbb{\R}^{p}\right\} ,
\]
and the nuisance tangent space for the proportional hazards model
(\ref{eq:ph-C}) is
\[
\widetilde{\Lambda}_{3}=\left\{ S_{\gamma_{C}}+\int h(u)\de M_{C}(u):\ h(u)\in\mathbb{\R}^{p}\right\} .
\]
Assuming that the treatment process and the censoring process can
not jump at the same time point, $\widetilde{\Lambda}_{1}$, $\widetilde{\Lambda}_{2}$,
and $\widetilde{\Lambda}_{3}$ are mutually orthogonal to each other.
Therefore, the nuisance tangent space for $\xi$ and the proportional
hazards models (\ref{eq:ph-V}) and (\ref{eq:ph-C}) is $\widetilde{\Lambda}=\widetilde{\Lambda}_{1}\oplus\widetilde{\Lambda}_{2}\oplus\widetilde{\Lambda}_{3}$.
The influence function for $\widehat{\psi}$ is
\begin{eqnarray*}
 &  & \widetilde{B}(\psi^{*},\xi^{*},M_{V},K_{C};F)\\
 & = & \Phi(\psi^{*},\xi^{*},M_{V},K_{C};F)-\prod\left\{ \Phi(\psi^{*},\xi^{*},M_{V},K_{C};F)\mid\widetilde{\Lambda}\right\} \\
 & = & \Phi(\psi^{*},\xi^{*},M_{V},K_{C};F)-E\left\{ \Phi(\psi^{*},\xi^{*},M_{V},K_{C};F)S_{\gamma_{V}}^{\T}\right\} E\left(S_{\gamma_{V}}S_{\gamma_{V}}^{\T}\right)^{-1}S_{\gamma_{V}}\\
 &  & -E\left\{ \Phi(\psi^{*},\xi^{*},M_{V},K_{C};F)S_{\gamma_{C}}^{\T}\right\} E\left(S_{\gamma_{C}}S_{\gamma_{C}}^{\T}\right)^{-1}S_{\gamma_{C}}\\
 &  & +\int\frac{E\left[G(\psi^{*},\xi^{*},M_{V};F)\exp\left\{ \gamma_{C}^{\T}g_{C}(u,\overline{H}_{u})\right\} \Delta/K_{C}(T\mid\overline{H}_{T})\right]}{E\left[\exp\left\{ \gamma_{C}^{\T}g_{C}(u,\overline{H}_{u})\right\} Y_{C}(u)\right]}\de M_{C}(u)\\
 &  & +\int\frac{E\left[G(\psi^{*},\xi^{*},M_{V};F)\exp\left\{ \gamma_{V}^{\T}g_{V}(u,\overline{H}_{u})\right\} \Delta/K_{C}(T\mid\overline{H}_{T})\right]}{E\left[\exp\left\{ \gamma_{V}^{\T}g_{V}(u,\overline{H}_{u})\right\} Y_{V}(u)\right]}\de M_{V}(u).
\end{eqnarray*}

\end{proof} 

\section{The Cox marginal structural model approach: $\widehat{\psi}_{\msm}$}

The Cox marginal structural model approach assumes that the potential
failure time under $\overline{a}_{T}$ follows a Cox proportional
hazards model with the hazard rate at $t$ as $\lambda_{0}(t)\exp(\psi^{*}a_{t})$.

If all potential failure times were observed for all subjects, one
can fit a Cox proportional hazards model with the time-varying covariate
$a_{t}$ to obtain a consistent estimator of $\psi^{*}$. However,
not all potential outcomes are observed for a particular subject.
To obtain a consistent estimator based on the actual observed data,
the key step is to construct time-dependent inverse probability of
treatment weights for all subjects and weight their contributions
so that they mimic the contributions had all potential outcomes been
observed.

From the hazard of treatment discontinuation $\lambda_{V}\left(t\mid\overline{H}_{t}\right)$
defined in (\ref{eq:UNC}), denote 
\begin{equation}
K_{V}\left(t\mid\overline{H}_{t}\right)=\exp\left\{ -\int_{0}^{t}\lambda_{V}(u\mid\overline{H}_{u})\de u\right\} \label{(11.3)}
\end{equation}
and 
\begin{equation}
f_{V}\left(t\mid\overline{H}_{t}\right)=\lambda_{V}\left(t\mid\overline{H}_{t}\right)K_{V}\left(t\mid\overline{H}_{t}\right).\label{(11.4)}
\end{equation}
For ease of notation, denote $K_{V}(t)=K_{V}\left(t\mid\overline{H}_{t}\right)$
and $f_{V}(t)=f_{V}\left(t\mid\overline{H}_{t}\right)$ for shorthand.
These can be viewed as the probability of not having discontinued
before time $t$ and the probability of discontinuing at time $[t,t+\de t)$,
respectively.

Consider subjects who were are at risk at time $t$. We consider two
subsets of individuals: group (a) with $V\leq t$ and $\Gamma=1$
and group (b) with $V>t$. Specifically, we construct the time-dependent
inverse probability of treatment weight as\textcolor{black}{
\begin{equation}
\omega(t)=\begin{cases}
\theta(V)/f_{V}(V), & \text{if }V\leq t\text{ and }\Gamma=1,\\
\overline{\theta}(t)/K_{V}(t) & \text{if }V>t,
\end{cases}\label{(13.1)}
\end{equation}
}where $\theta(t)$ and $\overline{\theta}(t)=\int_{t}^{\infty}\theta(u)\de u$
serve as the stabilized weights \citep{hernan2000marginal}. Following
\citep{yang2018modeling}, one can consider $\theta(t)=\lambda_{V,0}(t)\exp\left\{ -\int_{0}^{t}\lambda_{V,0}(u)\de u\right\} $.
In the presence of censoring, let $\omega(t)$ be a product of (\ref{(13.1)})
and the inverse of censoring probability $\Delta/K_{C}\left(T\mid\overline{H}_{T}\right)$.
One can estimate the weights by replacing the unknown quantities with
their estimates following Steps 1 and 2 in $\mathsection$ \ref{subsec:algorithm}.

Finally, we obtain $\widehat{\psi}_{\msm}$ by fitting a Cox proportional
hazards model with the time-varying covariate $A_{t}$ with the time-dependent
weight $\omega(t)$ using the standard software; e.g., the function
``coxph'' in R with the weighting argument.

\section{The discrete-time g-estimator: $\widehat{\psi}_{\disc}$}

The existing framework for fitting the structural failure time model
is using a discrete time points setting which requires manually discretizing
the data. We disretize the timeline into equally-spaced time points
from $0$ to the maximum follow up $\tau$, denoted as $0=t_{0}<t_{1}<\cdots<t_{K}=\tau$.
For $m\geq1$, at the $m$th time point $t_{m}$, let $A_{t_{m}}$
be the indicator of whether the treatment is received at $t_{m}$,
let $L_{t_{m}}$ be the the average of $L_{t}$ from $t_{m-1}\leq t\leq t_{m}$,
let $H_{t_{m}}$ be the vector of $A_{t_{m}-1}$ and $L_{t_{m}}$,
and finally let $\overline{H}_{t_{m}}$ be $\{H_{0},\ldots,H_{t_{m}}\}$.
With observations at discrete time points, $\de N_{T}(t_{m})$ becomes
the binary treatment indicator $A_{t_{m}}$, $\lambda_{T}(u\mid\overline{H}_{u})Y_{T}(u)\de u$
becomes the propensity score $E(A_{t_{m}}\mid\overline{H}_{t_{m}},\overline{A}_{t_{m}-1}=\overline{0})$,
and the integral in (\ref{eq:ee4}) becomes the summation from $m=1$
to $K$. As a result, in the absence of censoring, (\ref{eq:ee4})
simplifies to the existing estimating equation for structural nested
failure time models \citep{hernan2005structural}. Following \citep{hernan2005structural},
one can estimate the propensity score by the pooled logistic regression
model with baseline and time-dependent covariates. In the presence
of censoring, one can estimate the censoring probability by the pooled
logistic regression model with baseline and time-dependent. The g-estimator
$\widehat{\psi}_{\disc}$ of $\psi^{*}$ solves the estimating equation
(\ref{eq:IPCW2}) with observations at discrete time points.

\section{Details and additional results in the simulation}

In this section, we present details for the Jackknife method for variance
estimation and additional simulation results to assess the impact
of misspecification of the censoring model and the treatment effect
model.

The Jackknife method entails dividing the subjects into exclusive
and exhaustive subgroups, creating replicate datasets by deleting
one group at a time, and applying the same estimation procedure to
obtain the replicates of $\widehat{\psi}$. The variance estimator
is $\widehat{V}(\widehat{\psi})=G^{-1}(G-1)\sum_{k=1}^{G}\left(\widehat{\psi}^{(k)}-\widehat{\psi}\right)^{2},$
where $G$ is the number of subgroups, and $\widehat{\psi}^{(k)}$
is the $k$th replicate of $\widehat{\psi}.$

We now focus on the scenario 1 of the simulation study in $\mathsection$
\ref{sec:Simulation-Study}. First in setting 1, to illustrate the
impact of misspecification of the censoring model, for all estimators,
we consider an incorrect independent censoring mechanism for fitting
the censoring model in the sense that the censoring indicator is independent
of all other variables. Second in setting 2, to illustrate the impact
of misspecification of the treatment effect model, we now generate
the failure time, $T$, according to a structural failure time model
$U\sim\int_{0}^{T}\exp(\psi^{*}A_{u}+0.5X_{0})\de u$. All estimators
are the same as in $\mathsection$ \ref{sec:Simulation-Study}.

Table \ref{tab:Results-supp} summarizes the simulation results with
$n=1,000$. In setting 1 when the censoring model is misspecified,
the proposed estimators have larger biases compared to the results
when the censoring model is correctly specified as in Table \ref{tab:Results1}.
In setting 2 when the treatment effect model is misspecified, the
proposed estimators also have increased biases compared to the results
when the treatment effect model is correctly specified as in Table
\ref{tab:Results1}. The coverage rates are off the nominal coverage
in most of cases.

\begin{table}
\caption{\label{tab:Results-supp}Simulation results: bias, standard deviation,
root mean squared error, and coverage rate of $95\%$ confidence intervals
for $\exp(\psi^{*})$ over $1,000$ simulated datasets: Setting 1
where the censoring model is misspecified, and Setting 2 where the
treatment effect model is misspecified}

\centering{}\resizebox{\textwidth}{!}{%
\begin{tabular}{cccccccccccc}
 &  &  & \multicolumn{3}{c}{$\psi{}^{*}=-0.5$} & \multicolumn{3}{c}{$\psi{}^{*}=0$} & \multicolumn{3}{c}{$\psi{}^{*}=0.5$}\tabularnewline
 &  &  & Bias & S.E. & C.R. & Bias & S.E. & C.R. & Bias & S.E. & C.R.\tabularnewline
\multirow{2}{*}{} & \multirow{2}{*}{$\widehat{\psi}_{\naive}$} & $c$ & 0.06 & 0.048 & 76.8 & 0.02 & 0.069 & 95.6 & -0.06 & 0.112 & 92.4\tabularnewline
 &  & $c^{\opt}$ & 0.05 & 0.043 & 78.4 & 0.02 & 0.063 & 95.0 & -0.05 & 0.107 & 91.8\tabularnewline
\multirow{2}{*}{Setting 1} & \multirow{2}{*}{$\widehat{\psi}_{\ipcw}$} & $c$ & 0.14 & 0.086 & 68.2 & 0.04 & 0.103 & 97.4 & -0.13 & 0.153 & 82.2\tabularnewline
 &  & $c^{\opt}$ & 0.13 & 0.072 & 61.6 & 0.04 & 0.089 & 96.8 & -0.12 & 0.136 & 81.4\tabularnewline
\multirow{2}{*}{} & \multirow{2}{*}{$\widehat{\psi}_{\dr}$} & $c$ & 0.04 & 0.050 & 88.4 & 0.01 & 0.070 & 95.0 & -0.04 & 0.116 & 94.0\tabularnewline
 &  & $c^{\opt}$ & 0.04 & 0.048 & 89.0 & 0.01 & 0.069 & 95.2 & -0.03 & 0.115 & 93.2\tabularnewline
 & $\widehat{\psi}_{\msm}$ &  & 0.04 & 0.045 & 92.6 & 0.02 & 0.073 & 96.8 & -0.04 & 0.135 & 94.6\tabularnewline
 & $\widehat{\psi}_{\disc}$ &  & -0.40 & 0.035 & 0.0 & -0.65 & 0.046 & 0.0 & -1.09 & 0.070 & 0.0\tabularnewline
 &  &  &  &  &  &  &  &  &  &  & \tabularnewline
\multirow{2}{*}{} & \multirow{2}{*}{$\widehat{\psi}_{\naive}$} & $c$ & 0.05 & 0.049 & 87.0 & -0.01 & 0.069 & 94.8 & -0.12 & 0.108 & 80.2\tabularnewline
 &  & $c^{\opt}$ & 0.03 & 0.044 & 90.8 & -0.03 & 0.064 & 92.2 & -0.14 & 0.100 & 73.4\tabularnewline
\multirow{2}{*}{Setting 2} & \multirow{2}{*}{$\widehat{\psi}_{\ipcw}$} & $c$ & -0.02 & 0.088 & 95.6 & -0.05 & 0.120 & 94.6 & -0.10 & 0.174 & 90.8\tabularnewline
 &  & $c^{\opt}$ & -0.04 & 0.070 & 93.4 & -0.07 & 0.096 & 93.2 & -0.12 & 0.136 & 86.2\tabularnewline
\multirow{2}{*}{} & \multirow{2}{*}{$\widehat{\psi}_{\dr}$} & $c$ & -0.02 & 0.053 & 94.2 & -0.03 & 0.079 & 91.0 & -0.08 & 0.122 & 90.6\tabularnewline
 &  & $c^{\opt}$ & -0.02 & 0.050 & 91.6 & -0.05 & 0.073 & 89.4 & -0.11 & 0.115 & 84.8\tabularnewline
 & $\widehat{\psi}_{\msm}$ &  & -0.01 & 0.049 & 96.6 & -0.03 & 0.082 & 92.2 & -0.09 & 0.134 & 91.6\tabularnewline
 & $\widehat{\psi}_{\disc}$ &  & -0.38 & 0.041 & 0.0 & -0.63 & 0.053 & 0.0 & -1.06 & 0.085 & 0.2\tabularnewline
\end{tabular}}
\end{table}

\section{Nuisance models in the application}

In this section, we provide details for fitting the nuisance models
in the application. To build a model for $\lambda_{V}(t\mid\overline{H}_{t})$
in (\ref{eq:UNC}), we consider the baseline covariates $X$, including
age, gender, race, site, country, and other $25$ baseline health
outcome measures. For each categorical variable, we create dummy variables.
This leads to $99$ baseline variables. We first fit a Cox proportional
hazards model for $\lambda_{V}(t\mid\overline{H}_{t})$ to the data
including the baseline variables with a $l_{1}$ penalty. In fitting
the model, we select the tuning parameter using $10$-fold cross-validation.
The final proportional hazards model includes the selected baseline
terms and all time-dependent covariates $L_{t}$, including indicators
of bleeding, haemorrhagic stroke, and left atrial appendage procedures
associated with permanent discontinuation and outcomes. To build a
model for $\lambda_{C}(t\mid\overline{H}_{t})$ in (\ref{eq:censoring}),
we consider the same procedure for $\lambda_{V}(t\mid\overline{H}_{t})$.
This is because the decision to re-start treatment was left to the
patient and physician, and the resulting censoring may depend on the
patient's characteristics and evolving disease status. To estimate
$E\left\{ U(\psi)\mid\overline{H}_{0}\right\} $, we regress $\widehat{K}_{C}\left(T\mid\overline{H}_{T}\right)^{-1}\Delta U(\psi)$
on $X$ with a $l_{1}$ penalty.
\end{document}